\documentclass[a4paper, USenglish, cleveref, autoref, thm-restate]{lipics-v2021}

\hideLIPIcs  

\title{Forbidden Patterns in Mixed Linear Layouts}

\author{Deborah Haun}{Karlsruhe Institute of Technology, Germany}{deborah.haun@kit.edu}{https://orcid.org/0009-0004-2365-0804}{}
\author{Laura Merker}{Karlsruhe Institute of Technology, Germany}{laura.merker2@kit.edu}{https://orcid.org/0000-0003-1961-4531}{}
\author{Sergey Pupyrev}{Menlo Park, CA, USA}{spupyrev@gmail.com}{https://orcid.org/0000-0003-4089-673X}{}

\authorrunning{D. Haun, L. Merker, and S. Pupyrev} 

\Copyright{Deborah Haun, Laura Merker, and Sergey Pupyrev} 



\relatedversion{An extended abstract of this paper appeared in the proceedings of the 42nd International Symposium on Theoretical Aspects of Computer Science (STACS 2025) \cite{HMP24_STACS}}

\acknowledgements{We thank Torsten Ueckerdt for valuable feedback to a preliminary version of this paper.}

\nolinenumbers 

\usepackage{mathtools}
\usepackage{bm}
\usepackage{xspace}
\usepackage{csquotes}
\usepackage[nocompress,noadjust]{cite}
\usepackage[dvipsnames,cmyk]{xcolor}

\newtheorem{open}[theorem]{Open Problem}

\graphicspath{{pics/}}

\definecolor{defblue}{rgb}{0.121,0.47,0.705}
\newcommand{\df}[1]{{\color{defblue}\it #1}}
\renewcommand{\emph}[1]{\df{#1}}

\newcommand{\Oh}{{\ensuremath{\mathcal{O}}}\xspace}
\newcommand{\NP}{{\textrm{NP}}\xspace}

\newcommand{\Ph}{{\ensuremath{\mathcal{P}}}\xspace}

\newcommand{\Ih}{{\ensuremath{\mathcal{I}}}\xspace}
\newcommand{\Ah}{{\ensuremath{\mathcal{A}}}\xspace}
\newcommand{\Ch}{{\ensuremath{\mathcal{C}}}\xspace}

\DeclareMathOperator{\sn}{sn}

\DeclareMathOperator{\qn}{qn}

\DeclareMathOperator{\mn}{mn}

\DeclareRobustCommand{\diamondtimes}{%
	\mathbin{\text{\rotatebox[origin=c]{45}{$\boxplus$}}}%
}

\newcommand{\mathtitle}[2]{\texorpdfstring{\ensuremath{\bm{#1}}}{#2}}

\newcommand{\dexp}{7}
\newcommand{\sqf}{\ensuremath{k}\xspace}

\begin{document}

\maketitle

\begin{abstract}
An ordered graph is a graph with a total order over its vertices.
A linear layout of an ordered graph is a partition of the edges into sets of either non-crossing edges, called stacks, or non-nesting edges, called queues. The stack (queue) number of an ordered graph is the minimum number of required stacks (queues). Mixed linear layouts combine these layouts by allowing each set of edges to form either a stack or a queue. The minimum number of stacks plus queues is called the mixed page number. It is well known that ordered graphs with small stack number are characterized, up to a function, by the absence of large twists (that is, pairwise crossing edges).
Similarly, ordered graphs with small queue number are characterized by the absence of large rainbows (that is, pairwise nesting edges). However, no such characterization via forbidden patterns is known for mixed linear layouts.

We address this gap by introducing patterns similar to twists and rainbows, which we call thick patterns; such patterns allow a characterization, again up to a function, of mixed linear layouts of bounded-degree graphs. That is, we show that a family of ordered graphs with bounded maximum degree has bounded mixed page number if and only if the size of the largest thick pattern is bounded. In addition, we investigate an exact characterization of ordered graphs whose mixed page number equals a fixed integer $ k $ via a finite set of forbidden patterns.
We show that for every $ k \ge 2 $, there is no such characterization, which supports the nature of our first result.
\end{abstract}

\clearpage

\section{Introduction}\label{sec:intro}

An \df{ordered graph} is a graph given with a fixed linear vertex order, $ \prec $.
A \df{linear layout} of an ordered graph is a partition of its edges such that each part satisfies certain requirements with respect to the order.
In a \df{stack layout} each part, also called a \df{stack}, is required to be crossing-free with respect to $ \prec $, that is, two edges in the same stack may not have alternating endpoints.
A \df{queue layout} is a \enquote{dual} concept which forbids two edges to nest in the same part, called a \df{queue};
that is, if $ u \prec x \prec y \prec v $, then edges $ uv $ and $ xy $ must be in different queues.
The two concepts are generalized in \emph{mixed linear layouts}, where each part (called a \df{page} then) may either be a stack or a queue.
A linear layout using $ s $ stacks and/or $ q $ queues is called pure \emph{$ s $-stack}, pure \emph{$ q $-queue}, and
mixed \emph{$ s $-stack $ q $-queue}, respectively.
In all three cases, the objective is to minimize the number of parts.
The \df{stack number} $ \sn(G) $ (\df{queue number} $ \qn(G) $, \df{mixed page number} $ \mn(G) $) of an ordered graph $ G $ is the smallest $ k $ such that there is a stack (queue, mixed) layout with at most $ k $ stacks (queues, pages).

Stack layouts and queue layouts are well understood and a rich collection of tools has been developed. Most notably, the product structure theory~\cite{DJMMUW20}, layered path decompositions~\cite{DMY21,BDDEW19}, track layouts~\cite{DPW04}, and different kinds of $ H $-partitions that were used successfully both for queue layouts~\cite{DJMMUW20, HW24, FKMPR23, BGR23} and stack layouts~\cite{JMU23}.
A fundamental technique in this context is a characterization of stack and queue layouts via forbidden ordered patterns.
Formally, a \df{pattern} is an ordered graph with at least one edge. The \df{size} of a pattern is the number of edges.
An ordered graph $(G, \prec_1)$
\df{contains} a pattern $(H, \prec_2)$ if $H$ is a subgraph of $G$ and $\prec_2$ is a suborder of $\prec_1$;
otherwise, the graph \df{avoids} the pattern.
A \df{$k$-twist} denotes a set of $k$ pairwise crossing edges with respect to some vertex order, and a \df{$k$-rainbow} is a set of $k$ pairwise nesting edges, where symbol $ k $ can be omitted if not needed.
We define a \df{graph parameter} to be a function assigning a non-negative integer to every graph.
A parameter $p$ is \df{bounded} for a family of graphs if there is a constant $ c $ such that $p(G) \leq c $ for every graph $ G $ of the family.
Now, the characterization of stack and queue layouts can be formulated 
using the largest twist or rainbow.\footnote{We remark that the largest pattern is not necessarily unique. However, we often write \textit{the} largest pattern for improved readability.}

\begin{theorem}[\cite{Gya85,Dav22}]\label{thm:sn_twists}
    A family $\mathcal{G}$ of ordered graphs has bounded stack number if and only if 
    the size of the largest twist is bounded.
\end{theorem}

\begin{theorem}[\cite{HR92}]\label{thm:qn_rainbows}
   	A family $\mathcal{G}$ of ordered graphs has bounded queue number if and only if 
    the size of the largest rainbow is bounded.
\end{theorem}

These theorems are useful both for upper bounds (as an explicit assignment of edges to stacks or queues is not needed)
and for lower bounds (as a tedious case distinction which edge could go to which stack or queue gets superfluous).
Unfortunately, no analogous~characterization is known for mixed linear layouts.
As a consequence, all known results on mixed linear layouts rely on an explicit assignment of edges to stacks and queues.
We aim to close the gap and find a set of patterns whose absence characterizes mixed linear layouts similarly to twists and rainbows in \cref{thm:sn_twists,thm:qn_rainbows}.
The main question we study is as follows.


\begin{open}\label{op:mn_characterization}
    Do there exist finite sets, $\Ph_1, \Ph_2, \dots $, of patterns
    and a binding function, $f: \mathbb{N} \rightarrow \mathbb{N}$, such that
    for all $ k \in \mathbb{N} $ and
    every ordered graph, $G$, the following holds:
    \begin{itemize}
        \item $\mn(G) < f(k)$ if $G$ avoids all patterns in $\Ph_k$, and
        \item $\mn(G) \ge k$ if $G$ contains some pattern in $\Ph_k$?
    \end{itemize}
\end{open}

Note that for pure stack/queue layouts, the answer to the question is positive, since
the corresponding sets $\Ph_k$ of forbidden patterns (also called \df{obstruction sets})
contain a single element, a $k$-twist and a $k$-rainbow, respectively.
Indeed, a more technical formulation of \cref{thm:sn_twists} would be that for every ordered graph $ G $ and every $ k \in \mathbb{N} $, it holds that the stack number of $ G $ is less than $14 k \log k $ if $ G $ avoids a $ k $-twist, and is at least $ k $ if $ G $ contains a $ k $-twist~\cite{Gya85,Dav22}.
We remark that in contrast to stack layouts, queue layouts even admit the identity as binding function~\cite{HR92}.

On the other hand, an affirmative answer to \cref{op:mn_characterization} would give a theorem of the form of \cref{thm:qn_rainbows,thm:sn_twists}: There is a set \Ph of patterns, namely the union of $ \Ph_1, \Ph_2, \dots $, such that a family of ordered graphs has bounded mixed page number if and only if the size of the largest pattern of \Ph that occurs in a graph of the family is bounded.
To see the implication, note that there are only finitely many patterns of a certain size, so if the size of the largest occurring pattern of \Ph is bounded by a constant $c$, then the first property of \cref{op:mn_characterization} gives that the mixed page number is bounded, where $k$ is chosen maximal such that $\Ph_{k-1}$ contains a pattern of size $c$.
And conversely, a family whose mixed page number is at most $ k $ contains no patterns of \Ph that are larger than the largest pattern in $ \Ph_1, \dots, \Ph_k $ by the second property.

To complement \cref{thm:sn_twists,thm:qn_rainbows}, we answer \cref{op:mn_characterization} affirmatively for mixed linear layouts of bounded-degree graphs and provide negative results for an exact characterization with the identity binding function.

The remainder of this section is organized as follows.
First, we present our main results and describe technical contributions with concrete bounds
that guide through the subsequent sections.
Then, we relate our findings to the state-of-the-art. Finally, we discuss linear layouts from
various perspectives targeted to readers not familiar with the topic.

\subsection{Main Results}\label{sec:main_results}

We resolve \cref{op:mn_characterization} positively for bounded-degree ordered graphs.
For an explicit description of the set $\Ph_k$ of patterns, we define a \df{$k$-thick pattern} to be obtained either
from a $ k $-twist by replacing each edge by a $ k $-rainbow, or from a $ k $-rainbow by replacing each edge by a $ k $-twist; see \cref{fig:intro_forbidden_patterns}.%
\footnote{We emphasize that a $k$-thick pattern contains exactly $k^2$ edges, and hence, has size $k \times k$}
Refer to \cref{sec:thick} for a more elaborate introduction.

\begin{figure}
	\centering
	\includegraphics[page=3, width=0.48\textwidth]{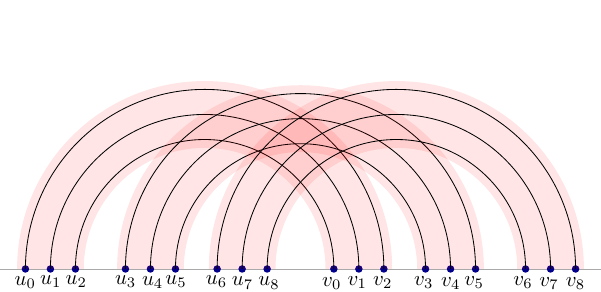}
	\hfill
	\includegraphics[page=4, width=0.48\textwidth]{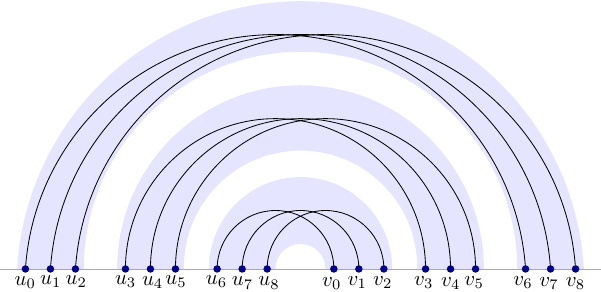}
	\caption{The two 3-thick patterns: three pairwise crossing 3-rainbows (left) and three pairwise nesting 3-twists (right)}
	\label{fig:intro_forbidden_patterns}
\end{figure}

\begin{theorem}\label{main:thick}
   	A family $\mathcal{G}$ of ordered graphs with bounded maximum degree has bounded mixed page number
   	if and only if 
    the size of the largest thick pattern that occurs in a graph of $\mathcal{G}$ is bounded.
\end{theorem}

Note that the above theorem provides a rough characterization of bounded-degree ordered graphs with bounded mixed
page number based on the size of the largest thick pattern contained in the graph, just like \cref{thm:sn_twists,thm:qn_rainbows} do for the stack and queue number.
That is, the
largest thick pattern implies an upper bound on the mixed page number but not its exact value.
A more granular version of \cref{op:mn_characterization} would be the one in which
the binding function is the identity, that is, a characterization of
ordered graphs admitting mixed linear layouts with exactly $k$ pages for $k \in \mathbb{N}$.
As a negative result, we show that such an exact characterization is impossible, by presenting an infinite family of patterns
that needs to be included in an obstruction set for ordered \df{matchings} (graphs of maximum degree one)
for every $k \ge 2$.

\begin{theorem}\label{main:critical}
	For $ k \geq 2 $, let $\Ph_k$ be a set of patterns such that the mixed page number of an ordered matching is at most $ k $
    if and only if it avoids all of $\Ph_k$.
	Then $\Ph_k$ is infinite.
\end{theorem}

\subsection{Technical Contributions}
\label{sec:summary}

While answering \cref{op:mn_characterization}, one can optimize three objectives.
Firstly, the number of patterns should be small, as each of the patterns might be considered individually
in applications of the characterization.
Secondly, the binding function should be small to yield effective upper bounds on the mixed page number.
And thirdly, the results should hold for graph classes that are as general as possible.
In this subsection, we present explicit bounds on the mixed page number with different trade-offs between the three objectives.

As a base for subsequent results, we start with ordered matchings having a separated layout
and a relatively large obstruction set.
Here, a layout of a bipartite graph is called \emph{separated} if all vertices of one part precede all vertices of the
other part in the order.
An explicit description of forbidden patterns, denoted $\diamondtimes$-patterns, is provided in \cref{sec:diamond}.

\newcommand{\restatemarker}[1]{#1}

\begin{restatable}[name=\restatemarker{\cref{sec:diamond}}]{theorem}{FixedCharThm}
    \label{thm:fixed_char}
    Let $M$ be a matching with a separated layout. If the largest $\diamondtimes$-pattern
    in $M$ has size $k \times k $, then the mixed page number of $M$ is at least $k$ and at most~$2k$.
\end{restatable}

As illustrated by \cref{thm:sn_twists,thm:qn_rainbows}, the obstruction sets for pure stack and queue layouts consist of only
one pattern, namely a twist and a rainbow, respectively.
Close enough, we construct exactly two patterns (namely, the two thick patterns shown in \cref{fig:intro_forbidden_patterns}),
to play the same role for mixed linear layouts. The simplification comes at the cost of a worsened binding function.

\begin{restatable}[name=\restatemarker{\cref{sec:thick}}]{theorem}{ThickPatternThm}
    \label{thm:thick_pattern}
    Let $M$ be a matching with a separated layout.
    If the largest thick pattern in $M$ has size $k \times k$, then the mixed page number of $M$ is at least $k$ and at most $\Oh(k^7)$.
\end{restatable}

We complete the study of ordered matchings by generalizing results for separated layouts to arbitrary vertex orders.

\begin{restatable}[name=\restatemarker{\cref{sec:general_matching}}]{theorem}{GeneralMatchingThm}
    \label{thm:general_matching}
    Let $M$ be an ordered matching.
    If the largest thick pattern in $M$ has size $k \times k$, then the mixed page number of $M$ is
    at least $ k $ and at most $ \Oh(k^{10} \log^2 k) $.
\end{restatable}

To prove this, we work with so-called quotients of linear layouts that can capture the global structure of an (ordered) graph and show how to transfer linear layouts of a quotient to the initial graph (\cref{lem:quotient}), which might be of independent interest.

With Vizing's theorem~\cite{Viz65}, all our results on matchings generalize to bounded-degree graphs;
in particular we obtain the following bounds for \cref{main:thick}.

\begin{restatable}[name=\restatemarker{\cref{sec:general_matching}}]{theorem}{BoundedDegreeGraphThm}
    \label{thm:bounded_degree_graph}
    Let $G$ be an ordered graph with maximum degree $ \Delta $.
    If the largest thick pattern in $G$ has size $k \times k$, then the mixed page number of $G$ is
    at least $ k $ and at most $ \Oh(\Delta k^{10} \log^2 k) $.
\end{restatable}

Note that this is a specification of \cref{main:thick}, since the upper bound implies that if there is no large thick pattern, then the mixed page number is bounded; conversely, if the mixed page number is bounded, then there is no large thick pattern due to the lower bound.

To conclude the discussion on thick patterns and $\diamondtimes$-patterns, we show  in \cref{sec:complexity} that their recognition is computationally easier than determining the mixed page number (\cref{thm:diamond_fpt,thm:thick_poly}), which we show to be \NP-hard (\cref{thm:mn_hardness}).

\medskip

Our final set of results concerns an exact characterization of ordered graphs with mixed page number $k$;
that is, we aim to optimize the binding function. To this end, we consider $(s, q)$- and $k$-\df{critical} graphs
that are edge-minimal forbidden patterns for $s$-stack $q$-queue layouts and mixed $k$-page layouts, respectively.
We show that for certain separated layouts, the corresponding obstruction sets $\Ph_k$ are finite, which
positively resolves \cref{op:mn_characterization} for these cases.

\begin{restatable}[name=\restatemarker{\cref{sec:critical_separated}}]{theorem}{CriticalSeparatedThm}
    \label{thm:critical_sep}
    For separated layouts, there exists a finite number of
    \begin{enumerate}[(i)]
        \item bounded-degree $(s, q)$-critical and $k$-critical graphs for all $s,q,k \ge 1$, and
        \item $(1, 1)$-critical and $2$-critical graphs.
    \end{enumerate}

\end{restatable}

General (non-separated) mixed layouts, however, cannot be characterized by a finite obstruction set,
as the next theorem illustrates. Our construction provides an infinite (though not complete) set of
minimal matchings that need to be included in the obstruction set and thereby implies \cref{main:critical}.

\begin{restatable}[name=\restatemarker{\cref{sec:critical_general}}]{theorem}{CriticalGeneralThm}
    \label{thm:critical_nonsep}
    For non-separated layouts, there exists an infinite number of
    \begin{enumerate}[(i)]
        \item $(s, q)$-critical matchings for all $s \ge 2$, $q \ge 0$, and
        \item $k$-critical matchings for all $k \ge 2$.
    \end{enumerate}

\end{restatable}

\renewcommand{\restatemarker}[1]{}

\subsection{Related Work}\label{sec:related_work}

Building on earlier notions~\cite{K74,O73}, the concepts of stack and queue numbers were first investigated by Bernhart and Kainen~\cite{BK79} in 1979 and Heath and Rosenberg~\cite{HR92} in 1992, respectively. Over the last three decades, there has been an extensive research on the concepts leading to numerous results for various graph families, including planar graphs~\cite{ABGKP20, BGR23, BKKPRU20, dFdMP95, DJMMUW20, FKMPR23, Ove98, Yan89, Yan20}, $1$-planar graphs~\cite{BBKR17}, graphs with bounded genus~\cite{HI92,Mal94} or bounded treewidth~\cite{GH01}, and
bounded-degree graphs~\cite{BFGMMRU19,W08,DMW19}.
In addition, the stack and queue numbers of directed acyclic graphs~\cite{HPT99, NP23, JMU23}, where the vertex order must respect the orientation of the edges, have been studied fruitfully, e.g., for (planar) posets~\cite{NP89, KMU18, ABGKP23, FUW21, Pup23} and upward planar graphs~\cite{FFRV13, JMU22a, JMU23}.

Linear layouts are often analyzed via a machinery of forbidden patterns that can describe a family of graphs by
excluding certain (ordered) subgraphs.
In their seminal paper on queue layouts, Heath and Rosenberg~\cite{HR92} identified rainbows as a simple characterization for queue layouts. They showed that the queue number with respect to a fixed vertex order always equals the size of the largest rainbow; this result has often been used to derive bounds on the queue number~\cite{DJMMUW20, BGR23, ABGKP23, BFGMMRU19, DW04, FUW21, ABGKP20, Pup23, KMU18}.
Similarly, twists offer a natural characterization for stack layouts, although the stack number of an ordered graph does not always equal the size of the largest twist. While the size of the largest twist is clearly a lower bound on the stack number, Davies~\cite{Dav22} showed an upper bound of $\Oh(k \log k)$ on the stack number, where $k$ denotes the size of the largest twist. This very recent and asymptotically tight~\cite{KK97} bound, along with earlier larger upper bounds~\cite{Gya85,DM21,KK97,Cer07} has been useful for bounding the stack number of various graph classes~\cite{JMU23, NP23, JMU22a, FFRV13}.
For small values of $ k $, it is known that an order without a $2$-twist (that is, when $k=1$) corresponds to an outerplanar drawing of a graph, which is a $1$-stack layout.
For $k=2$ (that is, an order without a $3$-twist), five stacks are sufficient and sometimes necessary~\cite{K88,Age96};
for $k=3$, it is known that $19$ stacks suffice~\cite{Dav22}.

Although Heath and Rosenberg~\cite{HR92} already suggested to study mixed linear layouts back in 1992, specifically conjecturing that every planar graph has a $1$-stack $1$-queue layout, significant progress began only quite recently when Pupyrev~\cite{Pup18} disproved their conjecture. Subsequently, other papers have followed strengthening Pupyrev's result by showing that not even series-parallel~\cite{ABKM22} or planar bipartite graphs~\cite{FKMPR23} admit a mixed $1$-stack $1$-queue layout.
Further investigations into more general mixed linear layouts have often been proposed~\cite{BFGMMRU19, NP23, BKKPRU20}.
Results on this include complete and complete bipartite graphs~\cite{ABGKP22}, subdivisions~\cite{DW05,M20,Pup18}, and computational hardness results~\cite{CKN19}.
Very recently, Katheder, Kaufmann, Pupyrev, and Ueckerdt~\cite{KKPU24} related the queue, stack, and mixed page number to each other and asked for an improved understanding of graphs with small mixed page number both in the separated and in the general setting as, together with their results, this would fully reveal the connection between pure stack/queue and mixed layouts.
All these investigations of mixed layouts share a difficulty of analysis due to the lack of a simple characterization
similar to rainbows and twists for stack and queue layouts, which is the primary motivation for this paper.

Our investigations start with bipartite graphs having a separated layout, which is a setting that has been widely studied (sometimes, implicitly)
under different names, such as \df{2-track layouts}, \df{2-layer drawings}, or \df{partitioned drawings}~\cite{CSW04,DW05,ABGKP22,ALFS20,Woo22,DDEL11,EW94,DPW04,EW94b,Nag05,Sud04,Pem92}.
Initially, we further narrow down our focus to matchings with a separated layout.
Separated matchings can also be viewed
as permutations of the right endpoints of the edges, relative to the order of the left endpoints.
In this context, the mixed page number of an ordered matching equals the minimum number of
monotone subsequences into which the permutation can be partitioned.
This particular question has been investigated in the past:
Permutations that can be partitioned into a fixed number of monotone subsequences, and
thus matchings with a fixed mixed page number, are
characterized by a finite, though potentially very large, set of forbidden subsequences~\cite{KSW96,FH06,Wa21}.

Finally, we mention a related topic on the extremal theory of ordered graphs~\cite{Tar19,PT06}, which has been studied
from the perspective of $0$-$1$-matrices in the case of bipartite graphs with a separated layout~\cite{Tar05,JJMM24,MT04,FH92,Gen21}.
In this field, the focus is on determining how many edges (or $1$-entries) an ordered graph (or a $0$-$1$-matrix) can have while avoiding a specified family of patterns (or submatrices).
Note that the maximum number of edges on a page of a linear layout is upper-bounded by about twice the number of vertices.
Thus, a necessary (but clearly not sufficient) condition for a set of patterns to characterize a fixed mixed page number is that these patterns must
enforce the number of edges to be linear in the number of vertices.
This necessary condition is met by the F\"{u}redi-Hajnal conjecture~\cite{FH92},
which was first proven by Marcus and Tardos~\cite{MT04}.
That is, in the separated setting it holds that, for every pattern $P$, excluding $P$ yields a linear number of edges. 


\subsection{Connections to Other Fields}\label{sec:connections}

In this section, we review linear layouts from different perspectives and connect them to related concepts.

\subparagraph{Data Structure Perspective.}
Stack layouts and queue layouts capture how well an ordered graph can be processed by stacks and queues, which can be best seen in the case of an ordered matching.
Here, the vertices are considered in the given order, and an edge is pushed into the data structure when its first endpoint is reached, and it is popped when its second endpoint is reached.
With a stack, an ordered matching can be processed if and only if the edges obey a first-in-last-out order, that is, if and only if no two edges have alternating endpoints.
Thus, an ordered matching can be processed by $ k $ stacks if and only if its stack number is at most $ k $.
Similarly, an ordered matching can be processed with a queue if and only if the edges obey a first-in-first-out order, which is equivalent to having no two nesting edges.
Again, an ordered matching can be processed by $ k $ queues if and only if its queue number is at most $ k $.
Now if we allow to use both stacks and queues and minimize the total number of data structures, we naturally obtain the mixed page number.
Note that the restriction to matchings is not necessary here: For the order in which edges sharing an endpoint are pushed and popped, the order or reverse order of the other endpoints is used as a tie-breaker, depending on whether we have a stack or queue.
We remark that linear layouts are typically defined for \textit{unordered} graphs, where the task is to pick both a vertex order and an edge partition.
Translated to data structures, this asks whether a given graph admits a vertex order such that it can be processed by few stacks, queues, or both.
Altogether, investigating linear layouts improves our understanding of how the three fundamental data structures, namely graphs, stacks, and queues, interact with each other.

\subparagraph{Coloring Perspective.}
Finding a stack layout of an ordered graph is equivalent to coloring a circle graph~\cite{Dav22}.
A \emph{circle graph} is the intersection graph of chords of a circle, where two chords intersect if they cross but not if they share an endpoint.
To see the equivalence, consider a straight-line drawing of an ordered graph with the vertices forming a circle.
Now two chords cross if and only if their endpoints alternate in the linear order of the graph.
Hence, all findings on colorings of circle graphs~\cite{Gya85,DM21,KK97,Cer07,Un88,Un92,GJMP80} also apply to stack layouts and vice versa.
Most importantly, \cref{thm:sn_twists} is proved in the language of circle graphs by showing that they are $ \chi $-bounded, that is, their chromatic number is bounded by a function of their clique number~\cite{Dav22,Gya85}.
Note that a twist in an ordered graph corresponds to pairwise crossing chords, and similarly the size of the largest rainbow can be seen in the circle with its chords, up to a factor of 2.
For this, we say a set of chords is \emph{parallel} if the circle can be partitioned into two parts, one with one endpoint of each chord, and the other with the second endpoints in the reverse order.
Now if the largest rainbow is of size $ k $, then the largest set of parallel chords is at least $ k $ and at most $ 2k $.
Together, a mixed linear layout asks for a partition of the chords into two groups:
The first group is supposed to have only few pairwise crossing chords and corresponds to the set of stacks in the mixed linear layout.
And the second group should have only small sets of parallel chords and thereby corresponds to the set of queues.

\subparagraph{Upward Stack Number Perspective.}
One of the most prominent open questions in the field of linear layouts is whether planar posets
and upward planar graphs have bounded stack number~\cite{NP89,FFRV13,JMU22a}.
A directed acyclic graph is called \emph{upward planar} if it can be drawn in the plane such that the edges are crossing-free and $y$-monotone.
Despite intensive research in this direction~\cite{H93,FFRV13,BDDDMP23,MS09,BDMN23,NP23,JMU22a,JMU23}, this problem is still widely open. 
Hoping for a positive answer, bounding the mixed page number of an upward planar graph is an intermediate step before
answering the question for pure stack layouts.
In light of this, our results provide a tool to attack this problem.

\section{Separated Layouts of Matchings}
\label{sec:separated_matching}

We start with matchings having a separated layout, which is the base for proving \cref{main:thick} in \cref{sec:general_matching}.
Recall that a linear layout of a bipartite graph $G = (V_1 \cup V_2, E)$ is \df{separated}
if all vertices of $V_1$ precede all vertices of $V_2$ in the vertex order.
For separated layouts, there is a one-to-one mapping of the edges
of a bipartite graph to (a subset of) the points of the $|V_1| \times |V_2|$ integer grid, which we call
the \emph{grid representation}, where the points are arranged according to the linear layout as shown in \cref{fig:fbd_grid}.
The grid is also sometimes referred to as the \df{reduced adjacency matrix} of a bipartite graph.
In the case when $G$ is a matching, the matrix is a permutation matrix, i.e., the grid has exactly
one point in each column and each row.
Observe that the edges of a queue (stack) correspond to a monotonically increasing (decreasing)
sequence on the grid.

Next we give a linear upper bound on the mixed page number in \cref{sec:diamond} using a comparably large obstruction set of patterns having mixed page number $ k $.
We then reduce the number of patterns to 2 in \cref{sec:thick} and still obtain a polynomial dependency on~$ k $.
Note that this suffices for proving that a given class of graphs has bounded mixed page number but does not allow to compute the exact mixed page number.
This is a similar situation as we have for stack layouts, where having no $ k $-twist does not guarantee that the stack number is less than $ k $, but it does imply an upper bound of $ \Oh(k \log k) $~\cite{Dav22}.

\subsection{\mathtitle{\diamondtimes}{Diamond}-Patterns}
\label{sec:diamond}

Consider the grid representation of a bipartite graph $G$ with a fixed separated layout.
That is, one part of the bipartition is represented by columns, the other by rows, and each edge is described by an $ x $- and a $ y $-coordinate indicating the column, respectively the row, of its endpoints.
For two edges $e_1, e_2$ of $G$ we write $e_1 \nearrow e_2$ if $x(e_1) < x(e_2)$ and $y(e_1) < y(e_2)$.
Similarly, we write $e_1 \searrow e_2$ if $x(e_1) < x(e_2)$ and $y(e_1) > y(e_2)$.
Consider a set of $k^2$ edges indexed by $e_{i,j}$ for $1 \le i \le k$ and $1 \le j \le k$. We say that
the edges form a \df{$k$-$\diamondtimes$-pattern} if the following holds:
\begin{enumerate}[(i)]
	\item $e_{i,j} \nearrow e_{i, j+1}$ for every $1 \le i \le k$ and $1 \le j < k$, and
	\item $e_{i,j} \searrow e_{i+1, j}$ for every $1 \le i < k$ and $1 \le j \le k$.
\end{enumerate}

\begin{figure}
    \begin{subfigure}[b]{.24\linewidth}
        \center
        \includegraphics[page=1,width=\linewidth]{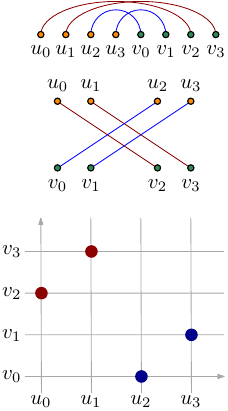}
        \caption{}
        \label{fig:fbd_grid_thick_rainbow}
    \end{subfigure}
    \hfill
    \begin{subfigure}[b]{.24\linewidth}
        \center
        \includegraphics[page=2,width=\linewidth]{pics/forbidden_subgrids}
        \caption{}
        \label{fig:fbd_grid_thick_twist}
    \end{subfigure}
    \hfill
    \begin{subfigure}[b]{.24\linewidth}
        \center
        \includegraphics[page=3,width=\linewidth]{pics/forbidden_subgrids}
        \caption{}
    \end{subfigure}
    \hfill
    \begin{subfigure}[b]{.24\linewidth}
        \center
        \includegraphics[page=4,width=\linewidth]{pics/forbidden_subgrids}
        \caption{}
    \end{subfigure}
    \caption{2-$\diamondtimes$-patterns for \cref{thm:fixed_char}}
    \label{fig:fbd_grid}
\end{figure}

That is, in the grid representation we have $ k $ increasing sequences of length $ k $, where the $ j $-th elements in each increasing sequence together form a decreasing sequence for each $ j = 1, \dots, k $.
We denote the \emph{size} of a $ k $-$ \diamondtimes $-pattern by $ k \times k $ to indicate that we have $ k^2 $ edges, where $ k $ is the parameter we are usually interested in.
Some $\diamondtimes$-patterns of size $ 2 \times 2 $ are illustrated in \cref{fig:fbd_grid}.
As it turns out, $ \diamondtimes $-patterns are deeply connected to mixed linear layouts.
Our main result on $ \diamondtimes $-patterns is summarized in the following theorem.

\FixedCharThm*

Before proving \cref{thm:fixed_char}, let us discuss $ \diamondtimes $-patterns from different perspectives.
First, observe that two edges $ e_1, e_2 $ with $ e_1 \nearrow e_2 $ have their left endpoints in the same order as their right endpoints, and as we consider separated layouts, they cross.
Thus, an increasing sequence corresponds to a twist in the linear layout.
Similarly, two edges $ e_1, e_2 $ with $ e_1 \searrow e_2 $ nest, and a decreasing set of edges forms a rainbow.
That is, the $ k^2 $ edges of a $ k $-$ \diamondtimes $-pattern form $ k $ twists of size $ k $, one for each index $ i = 1, \dots, k $, and at the same time, they form $ k $ rainbows of size $ k $, one for each index $ j = 1, \dots, k $.
Observe that each twist and each rainbow intersect in exactly one edge.
We refer to the top row of \cref{fig:fbd_grid} for an illustration.

Another perspective that we crucially use throughout the proof of \cref{thm:fixed_char} is that we interpret a
$ \diamondtimes $-pattern, or more generally any matching $ M $ with a separated vertex order, as a poset $ P(M) $ whose elements are the edges of $ M $.
Since $M$ is a matching, for every two elements, $e_1, e_2$ of $P(M)$ with $x(e_1) < x(e_2)$, it holds that either $e_1 \nearrow e_2$ or $e_1 \searrow e_2$.
In the former case, we set $e_1 < e_2$ in $P(M)$ and in the latter we make the two elements incomparable, for which we write $e_1 \parallel e_2$.
We call $ P(M) $ the \emph{poset of $ M $}.
Note that this indeed defines a poset as the relation $ < $ is the intersection of two linear orders, namely the $ x $-coordinates and the $ y $-coordinates in the grid representation.
Further observe that a chain of $ P(M) $ is increasing in the grid representation and corresponds to a twist in $ M $, and an antichain is decreasing in the grid representation and corresponds to a rainbow.
In particular, the \emph{height} of the poset, defined as the size of the largest chain, is equivalent to the size of the largest twist, and an antichain decomposition corresponds to a stack layout.
Similarly, the \emph{width} of the poset, that is, the size of the largest antichain, is the same as the size of the largest rainbow, and a chain decomposition corresponds to a queue layout.
In mixed linear layouts, however, we are allowed to use both stacks and queues, which translates to a decomposition of the poset into chains and antichains.
That is, to bound the mixed page number, we aim to cover each element of the poset by a chain or an antichain (or both) and thereby minimize the number of chains plus antichains.

\begin{observation}\label{obs:diamond_poset}
	A matching $ M $ with separated vertex order admits an $ s $-stack $ q $-queue layout if and only if its poset $ P(M) $ can be covered by $ s $ antichains plus $ q $ chains.
	In particular, the mixed page number $ \mn(M) $ is at most $ m $ if and only if $ P(M) $ can be covered by $ m $ chains and antichains.
\end{observation}

In the special case of a $ k $-$ \diamondtimes $-pattern, the poset consists of $ k^2 $ elements admitting both a chain partition of $ k $ chains of size $ k $ and an antichain partition of $ k $ antichains of size $ k $ such that each chain and each antichain share exactly one element.

Having this, the lower bound of \cref{thm:fixed_char} asks for a proof that $ P(M) $ cannot be covered by less than $ k $ chains and antichains.
In the following lemma, we show a slightly stronger statement.

\begin{samepage}
\begin{lemma}\label{lem:diamond_lb}
	Let $ M $ be a $ k $-$ \diamondtimes $-pattern and let $ P(M) $ be its poset.
	Then every decomposition of $ P(G) $ into a minimum number of chains and antichains consists either of $ k $ chains or of $ k $ antichains.
\end{lemma}
\end{samepage}

\begin{proof}
	Recall that for each $ i = 1, \dots, k $, the elements $ e_{i, 1}, \dots, e_{i, k} $ form a chain, and thus every antichain contains only one of them.
	Hence, every antichain, and symmetrical every chain, contains at most $ k $ elements.
	As we have $ k^2 $ elements, we need at least $ k $ (anti)chains.
	Note that the $ k $ (anti)chains need to be pairwise disjoint to cover all elements, which is only achieved with $ k $ chains or with $ k $ antichains, whereas every chain of size $ k $ and every antichain of size $ k $ share exactly one element.
\end{proof}

That is, $ \diamondtimes $-patterns are particularly hard graphs for mixed linear layouts as they do not profit from the possibility to mix chains and antichains, respectively queues and stacks, which makes them a suitable candidate for a characterization.
In the remainder of this subsection, we prove the upper bound of \cref{thm:fixed_char}, which confirms that these candidates are, up to a factor of 2, the only reason for a large mixed page number.

\begin{figure}
	\centering
	\includegraphics{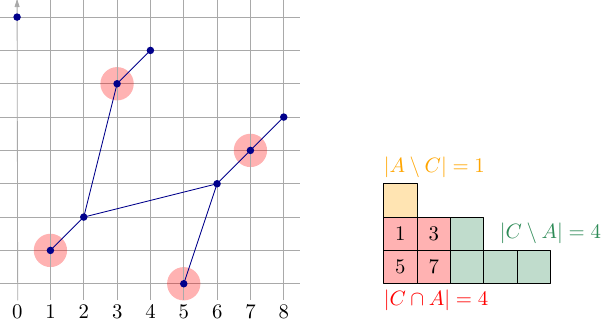}
	\caption{%
		Left: A grid representation of a matching $ M $ with edges indicating the comparabilities
        in the poset $ P(M)$ of $ M $.
		The 2-$ \diamondtimes $-pattern is highlighted red.
		Right:
        The Ferrers diagram of $ M $ showing the partition $ |P(M)| = 5 + 3 + 1 $ by rows of length 5, 3 and 1.
        The diagram expresses that one chain can cover five elements (bottom-most row), two chains can cover $ 5 + 3 = 8 $ elements (two bottom-most rows), and that three chains can cover all $ 5 + 3 + 1 = 9 $ elements (all three rows).
		Note, however, that eight elements cannot be covered with two chains of length 5 and 3.
		The largest square is of size $ 2 \times 2 $ and can be filled with elements of $ P(M) $ representing a $ 2 $-$ \diamondtimes $-pattern by \cref{lem:square_diamond}.
		$ C $ and $ A $ denote a maximum set of elements that can be covered by two chains, respectively two antichains, and the sizes of $ C \setminus A $ and $ A \setminus C $ are highlighted.
	}
	\label{fig:ferrer}
\end{figure}

The upper bound relies on insights into a Ferrers diagram, a standard combinatorial tool \cite{E96}, that represents how many elements can be covered by a certain number of chains, respectively antichains.
A Ferrers diagram represents a partition of an integer by left-aligned rows of cells, where a summand $ s $ is represented by a row of length $ s $ and the rows are ordered by non-increasing length from bottom to top.
See \cref{fig:ferrer} for the Ferrers diagram of 9 with the partition $ 5 + 3 + 1 = 9 $.
In our context, we partition the number of elements of a poset.

Roughly speaking, the Ferrers diagram of a poset $ P $ consists of $ |P| $ cells arranged in such a way that the number of cells in the first $ i $ rows is the number of elements in $ P $ that can be covered by $ i $ chains, and the number of cells in the first $ i $ columns equals the number of elements that can be covered by $ i $ antichains.
Formally, for a poset $ P $, we first consider two Ferrers diagrams of $ |P| $ with the following partitions.
For $ i \geq 0 $, let $c_i$ and $a_i$ denote the maximum number of elements of $ P $ that can be covered by $ i $ chains, respectively $ i $ antichains.
We call a set of $ i $ chains (antichains) whose union contains $ c_i $ ($a_i$) elements, a
\emph{maximum chain (antichain) $ i $-family}.
Let $ w $ and $ h $ denote the smallest $ i $ such that $ c_i = |P| $, respectively $ a_i = |P| $, that is, all elements can be covered by $ c_w $ chains or $ a_h $ antichains.
For the Ferrers diagrams we now use the differences between two consecutive $ c_i $, respectively $ a_i $.
More precisely, let $ F_c(P) $ denote the Ferrers diagram for the partition $ |P| = \sum_{i = 1}^{w} \hat{c}_i $, where $\hat{c}_i = c_i - c_{i-1}$.
Analogously, $ F_a(P) $ is the Ferrers diagram for the partition $ |P| = \sum_{i = 1}^{h} \hat{a}_i $, where $\hat{a}_i = a_i - a_{i-1}$.
For our upper bound, it is key that the two diagrams have the same shape up to mirroring, which is first proved by Greene~\cite{Gre76}.
For this, two Ferrers diagrams are called \emph{conjugate} if they are obtained from one another by flipping along the main diagonal, i.e., by making the rows to columns and vice versa.

\begin{theorem}[Greene~\cite{Gre76}]\label{thm:ferrer_conjugate}
	For every poset $ P $, the Ferrers diagrams $ F_c(P) $ and $ F_a(P) $ are conjugate.
\end{theorem}

That is, from now we only need to consider one of $ F_c(P) $ and $ F_a(P) $ as they have the same information, and call $ F(P) = F_c(P) $ the \emph{Ferrers diagram of $ P $}.
For a matching $ M $ and its poset $ P(M) $, we call $ F(P(M)) $ the \emph{Ferrers diagram of $ M $} and write $ F(M) $ for ease of presentation.
Observe that in $ F(P) $, the number of cells in the first $ i $ rows is the number of elements in $ P(M) $ that can be covered by $ i $ chains, and the number of cells in the first $ i $ columns equals the number of elements that can be covered by $ i $ antichains. 

We remark that this kind of Ferrers diagram, sometimes also called \emph{Greene's diagram}, is used by Felsner~\cite{Fel19} to show that $ P(M) $ can be covered by $ \Oh(\sqrt{n}) $ chains and antichains, which yields $ \mn(M) \in \Oh(\sqrt{n}) $.
For \cref{thm:fixed_char}, however, we want a bound depending on the size of the largest $ \diamondtimes $-pattern, so we need to identify these patterns in the Ferrers diagram.
The largest square in the Ferrers diagram (which we always assume to contain the bottom-left cell) is a natural point to start with, so let us first discuss why this is not a trivial task.
Recall that the Ferrers diagram only indicates the number of elements that can be covered by a certain number of (anti)chains but does not associate the cells with elements of the poset as this is not always possible.
For an example, consider \cref{fig:ferrer} where we have three rows of size 5, 3, and 1 but there is no chain decomposition with chains of lengths 5, 3, 1.
That is, there is no way of writing elements into all cells of the diagram such that elements in the same row form a chain.
Therefore, we cannot simply use the \enquote{elements in the $ k \times k $-square} to identify a
$ k $-$ \diamondtimes $-pattern but need a more elaborate argument.
Let us also mention that the largest square and the largest $ \diamondtimes $-pattern are not necessarily of the same size; see \cref{fig:square_diamond_reverse}.

Besides the translation of the linear layout problem to posets and Ferrers diagrams, the core of the proof of \cref{thm:fixed_char} is a counting argument that allows us to identify a $\diamondtimes$-pattern that is at least as large as the largest square in the Ferrers diagram.
Having this, we may use $ 2k $ (anti)chains, where $ k $ is the size of the largest square instead of the $ \diamondtimes $-pattern.
To do so, we make use of the following lemma.

\begin{lemma}\label{lem:ferrer_square}
	Let $ P $ be a poset and $ F(P) $ be its Ferrers diagram, where the largest square is of size $ k \times k $.
	Then for every maximum chain $ k $-family $ \Ch $ and every maximum antichain $ k $-family $ \Ah $, exactly $ k^2 $ elements are covered both by $ \Ch $ and $ \Ah $.
\end{lemma}

\begin{proof}
	First, each chain and each antichain share at most one element, and hence the $ k $ chains in $ \Ch $ and the $ k $ antichains in $ \Ah $ share at most $ k^2 $ elements.

	Now, let $ n $ be the number of elements of $ P $, and let $ C $ and $ A $ denote the set of elements covered by $ \Ch $, respectively $ \Ah $.
	Recall that $|C|$ counts the number of cells in the first $k$ rows of the Ferrers diagram $F(P)$, and $|A|$ counts the number of cells in the first $k$ columns.
	Since the largest square in $F(P)$ is of size $k \times k$, we count all $n$ cells of $F(P)$ and double-count exactly the cells in the $k \times k$-square.
	Thus, $| C | + | A | = n + k^2$.
	We conclude that
	$|C \cap A|
	= | C | + | A | - |C \cup A|
	\geq | C | + | A | - |P|
	= n + k^2 - n = k^2$.
\end{proof}

Having this, we are now able to identify $ \diamondtimes $-patterns, for which we use the following definitions.
Consider the grid representation $ \Gamma $ of a matching $ M $ and recall that a chain in the poset $ P(M) $ corresponds to an increasing path in $ \Gamma $.
Similarly, an antichain corresponds to a decreasing path, where the edges of the path represent incomparabilities.
See also \cref{fig:square_diamond_reverse} which contains some increasing paths (solid) and some decreasing paths (dashed).
We say that a family of chains (antichains) in $ P(M) $ is non-crossing if the corresponding paths do not cross in the grid representation.
Note that for every family of $ k $ chains, there is a family of $ k $ non-crossing chains covering the same elements, and symmetrically the same holds for antichains.
To see this, consider a family of chains with the minimum number of crossings.
If two chains $ C, C' $ cross, consider two crossing edges $ uv \in C $ and $ u'v' \in C' $ and replace $ C, C' $ by the chains that start like $ C, C' $ until $ u, u' $ and then continue from $ u $ with $ v' $, respectively from $ u' $ with $ v $.
This indeed yields chains since in the grid representation the paths are increasing from $ u, u' $ to the crossing, and from the crossing to $ v, v' $.
Hence, in the case of chains, both $ v $ and $ v' $ are to the top-right of $ u, u' $, respectively to the bottom-right in the case of two crossing antichains.
This decreases the number of crossings by 1 while covering the same elements, which contradicts the choice of the family.

\begin{lemma}\label{lem:square_diamond}
	Let $ M $ be a matching with a separated vertex order and let $ F(M) $ be its Ferrers diagram.
	If $ F(M) $ contains a square of size $ k \times k $, then $ M $ contains a $ \diamondtimes $-pattern of size $ k \times k $.
\end{lemma}

\begin{figure}
	\centering
	\includegraphics{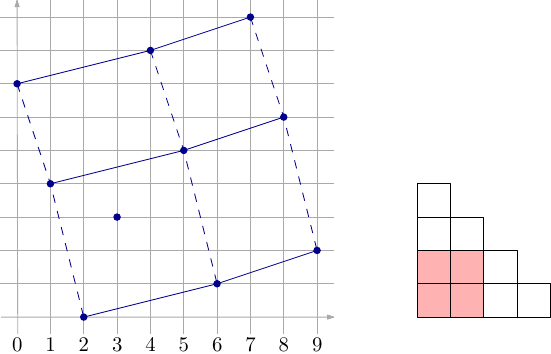}
	\caption{%
		Left: A grid representation of a matching, where some (in)comparabilities in the corresponding poset are shown with solid (dashed) edges.
		Right: The corresponding Ferrers diagram.
		The matching contains a $\diamondtimes$-pattern of size $ 3 \times 3 $ but its Ferrers diagram contains only a $ 2 \times 2 $-square.
		In this case, \cref{lem:ferrer_ub} gives a stronger upper bound of $ 2 \cdot 2 $ instead of $ 2 \cdot 3 $ as guaranteed by \cref{thm:fixed_char}.
	}
	\label{fig:square_diamond_reverse}
\end{figure}

Note that the reverse is not true; recall \cref{fig:square_diamond_reverse}.
We remark that \cref{lem:square_diamond} is also presented by Viennot~\cite{Vie84} together with a proof sketch.
As, to the best of our knowledge, the announced full version has never appeared, we give a proof for the sake of completeness.

\begin{proof}
	Let $ P(M) $ be the poset of $ M $ and let $ \Ch = \{C_1, \dots, C_k \} $ and $ \Ah = \{A_1, \dots, A_k \} $ denote a maximum chain (antichain) $ k $-family in $ P(M) $.
    By the observation above, we may assume without loss of generality that $ \Ch $ and $ \Ah $ are non-crossing and indexed from top to bottom, respectively left to right.
	Note that each chain and each antichain share at most one element.
	By \cref{lem:ferrer_square}, $ \Ch $ and $ \Ah $ share exactly $ k^2 $ elements, so one for each pair $ C_i, A_j $, which we denote by $ e_{i,j} $ according to their containment in $ C_i \cap A_j $.

	We show that the edges $e_{i,j}$ form a $\diamondtimes$-pattern of size $k \times k $. First, consider two edges $e_{i,j}, e_{i,j+1}$ for $1 \leq i \leq k$ and $1 \leq j < k$. Their corresponding poset elements are in the same chain $C_i$, and thus
	$e_{i,j} \nearrow e_{i,j+1}$ by construction of the poset $P(M)$.
	Conversely, for two edges $e_{i,j},e_{i+1,j}$ ($1 \leq i < k$ and $1 \leq j \leq k$), the corresponding poset elements are in the same antichain $A_j$, and thus  $e_{i,j} \searrow e_{i+1,j}$. So, the edges $e_{i,j}$ indeed form a $\diamondtimes$-pattern.
\end{proof}

Recall that our goal for the upper bound is to cover all elements by $ 2k $ chains and antichains, where $ k \times k $ is the size of the largest $ \diamondtimes $-pattern.
With \cref{lem:square_diamond}, we can equivalently bound the number of chains and antichains in terms of the largest square in the Ferrers diagram.

\begin{lemma}\label{lem:ferrer_ub}
	Let $ P $ be a poset and $ F(P) $ be its Ferrers diagram, where the largest square is of size $ k \times k $.
	Then for every maximum chain $ k $-family $ \Ch $ and every maximum antichain $ k $-family $ \Ah $, the union of $ \Ah $ and $ \Ch $ covers $ P $.
	In particular, $ P $ can be covered by $ k $ chains plus $ k $ antichains.
\end{lemma}

\begin{proof}
	Let $ n $ be the number of elements of $ P $, and let $ C $ and $ A $ denote the set of elements covered by $ \Ch $, respectively $ \Ah $.
	The sizes of $ C $ and $ A $ can be counted in the Ferrers diagram and add up to $ n + k^2 $.
	Now $ P $ can be partitioned into three sets:
	First, $ |C \cap A| = k^2 $ by \cref{lem:ferrer_square}.
	And second, we have the sets $ C \setminus A $ and $ A \setminus C $ containing $ |C| - k^2 $, respectively $ |A| - k^2 $ elements.
	This adds up to $ |C \cup A | = |C \setminus A | + |A \setminus C| + |C \cap A | = |C| - k^2 + |A| - k^2 + k^2 = |C| + |A| - k^2 = n + k^2 - k^2 = n $ elements, so $ \Ch $ and $ \Ah $ together cover all $ n $ elements of~$ P $.
\end{proof}

For an interpretation of the lemmas above, let us remark that even though we cannot write numbers in all cells of the Ferrers diagram, \cref{lem:ferrer_square,lem:square_diamond,lem:ferrer_ub} allow us to assign an element of $ P $ to each cell in the largest square and to associate the set of remaining cells to the right with $ C \setminus A $ and the remaining cells above the square with $ A \setminus C $; see also \cref{fig:ferrer}.

Now \cref{thm:fixed_char} is proved by putting everything together.

\begin{proof}[Proof of \cref{thm:fixed_char}]
	The lower bound is handled by \cref{lem:diamond_lb}.
	For the upper bound, let $ M $ be a matching with a separated vertex order, where the largest $ \diamondtimes $-pattern has size $ k \times k $ and let $ P(M) $ be the poset of $ M $.
	We aim to show that $ M $ can be covered by $ 2k $ pages.
	For this, first consider the Ferrers diagram $ F(M) $ of $ P(M)$, for which \cref{lem:square_diamond} shows that it does not contain a square of size $ (k+1) \times (k+1) $ as this would imply a $ \diamondtimes $-pattern of the same size.
	By \cref{lem:ferrer_ub}, $ P(M) $ can be covered by $ k $ chains plus $ k $ antichains.
	Finally \cref{obs:diamond_poset} shows that this implies that $ M $ admits a $ k $-stack $ k $-queue layout and thus has mixed page number at most $ 2k $.
\end{proof}

We remark that the Ferrers diagram 
can be computed in $ \Oh(n^{3/2}) $ time with the Robinson-Schensted algorithm~\cite{T22}.
From the proof of \cref{thm:fixed_char} we see that this translates to a 2-approximation algorithm for the mixed page number of matchings with a given separated vertex order.
Additionally, the chain and antichain $ k $-families can be computed in $ \Oh(k n \log n) $~\cite{FW98} where $ k $ is obtained from the largest square in the Ferrers diagram.
With $ k \leq \sqrt{n} $ this yields an $ \Oh(n^{3/2} \log n) $-time algorithm for $2k$-page assignment certifying the upper bound.
Note that having the $ k $-families, the proof of \cref{lem:square_diamond} shows how to obtain a $ k $-$\diamondtimes$-pattern, that is, we also get a certificate for the lower bound in the same time.

Finally, we show that \cref{thm:fixed_char} is tight.
First, the lower bound is indeed tight: there exist matchings with a separated layout for which the largest
$\diamondtimes$-pattern has size $k$ and whose mixed page number is exactly $k$.
For an example, we refer to the thick patterns defined in \cref{sec:thick} and illustrated in \cref{fig:thick_pattern}.
Additionally, we show that the upper bound of $2k$ can also be obtained.

\begin{lemma}\label{prop:diamond_tight}
	For every $k \geq 1$, there exists a matching with a separated layout that has mixed page number $2k$, whereas its largest $\diamondtimes$-pattern is of size $k \times k$.
\end{lemma}

\begin{proof}
	To construct a separated matching with the desired properties, we take the grid representation of a $k$-$\diamondtimes$-pattern and extend
	its increasing sequences to the top right and its decreasing sequences to the top left (\cref{fig:2k_is_tight}~left).
    Recall that an increasing sequence corresponds to a twist, respectively a chain, and can be covered by a single queue, while a decreasing sequence corresponds to a rainbow, respectively an antichain, and can be covered by a single stack.
    The first extension (green) has the purpose to enforce the usage of $k$ queues in any mixed linear layout with at most $2k$ pages, and similarly the second (blue) enforces $k$ stacks.


    For the concrete placement of the points, we need to take care that we place exactly one point into each column and row, in order to obtain a matching, without increasing the size of the largest $\diamondtimes$-pattern.
    The construction for $k=3$ is illustrated in \cref{fig:2k_is_tight}~(right).
    %
    %
    We start with $k$ decreasing $k$-sequences arranged to the top right of each other.
    Observe that this yields a $k$-$\diamondtimes$-pattern, namely a $k$-thick pattern as shown in \cref{fig:intro_forbidden_patterns}~(left) and \ref{fig:fbd_grid_thick_twist}.
    We extend this to the top right by adding $k(k-1)+1$ additional decreasing $k$-sequences, again to the top right of each other, see \cref{fig:2k_is_tight} (green).
    Then, above and to the left of all already placed points, we add $k(k-1)+1 $ increasing sequences, arranged to the top left of each other.
    Having all points placed, we observe that both the longest increasing sequence and the longest decreasing sequence are of length $ k + k(k-1)+1 = k^2 + 1 $ (green and black, respectively blue and black or green).

    %
    %

	From this it follows that $2k$ pages are indeed
	necessary to cover the constructed matching:
	On the one hand, $2k-1$ pages cover at most $(2k-1)(k^2+1) = 2k^3 - k^2 + 2k - 1$ edges, since every queue (increasing sequence) and every stack (decreasing sequence) contains at most $k^2 + 1$ edges.
	On the other hand, the total number of elements is $k^2 + 2 \cdot (k(k-1)+1) \cdot k = k^2 + 2k^3 - 2k^2 + 2k = 2k^3 - k^2 + 2k$, that is,
	at least one edge is uncovered.

    To finish the proof, it is left to show that we have not increased the size of the largest $\diamondtimes$-pattern when extending the initial $k$-thick pattern.
    Indeed, a $(k+1)$-$\diamondtimes$-pattern would imply $k+1$ maximal, pairwise disjoint increasing sequences of length at least $k+1$, but every increasing sequence with a point in the extension to the top left (blue) has length at most $k$, and the other points contain only $k$ such sequences.
    %
    %
    %
\end{proof}

\begin{figure}
	\includegraphics[page=2, scale=0.9]{pics/2k_is_tight}
	\hfill
	\includegraphics[page=1]{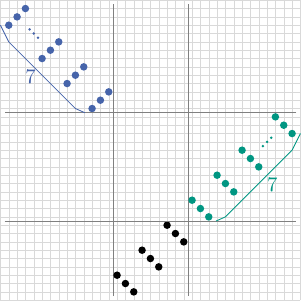}
	\caption{%
        Construction of a separated matching in grid presentation with mixed page number $2k$ that contains no $\diamondtimes$-pattern of size larger than $k \times k$.
		Left:~high-level idea of the construction: given a $k$-$\diamondtimes$-pattern (black) extend the increasing sequences to the top right (green) and the decreasing sequences to the top left (blue). 
        Right:~concrete construction for $k = 3$}
	\label{fig:2k_is_tight}
\end{figure}

\subsection{Connection to Twists and Rainbows}%
\label{sec:thick}

\cref{thm:fixed_char} provides a set of patterns that are the only reason for large mixed page number in separated matchings, up to a factor of~2.
In contrast to stack and queue layouts, where there is a single pattern (a twist and a rainbow, respectively),
the set of patterns is relatively large. Note that all $k$-$ \diamondtimes $-patterns need to be included in
an exact characterization of ordered graphs with mixed page number $ k $.
In this section, we describe a way of reducing the number of patterns in an obstruction set, at the cost of worsening
the bound on the mixed page number.
Specifically, we present two $ \diamondtimes $-patterns whose absence characterizes the mixed page number up to a polynomial factor.

As we may use both stacks and queues in mixed linear layouts, we combine twists and rainbows to obtain two new patterns for a characterization.
A \emph{$t$-thick $k$-twist} is obtained from a $k$-twist by replacing each edge by a $ t $-rainbow.
This yields $ k $ pairwise vertex-disjoint $t$-rainbows, where each edge nests with the edges belonging to the same rainbow but crosses the edges of all other rainbows; see \cref{fig:thick_twist}.
Similarly, a \emph{$t$-thick $k$-rainbow} is obtained from a $ k $-rainbow by replacing each edge by a $ t $-twist; see \cref{fig:thick_rainbow}.
In both cases, we call value $ t $  the \emph{thickness} and write the \emph{size} as $ k \times t $ to emphasize that we have $ k \cdot t $ edges organized as $ k $ smaller groups of size $ t $.
Throughout the paper, we often have $ t = k $ in which case we simply write $ k $-thick rainbow and $ k $-thick twist, respectively.
If we mean any of the two patterns, we say \emph{$ k $-thick pattern}.

\begin{figure}
	\begin{subfigure}[b]{\linewidth}
		\centering
		\includegraphics[page=5]{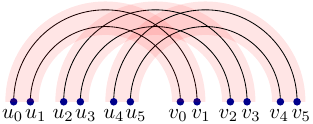}
		\hspace{4em}
		\includegraphics[page=6]{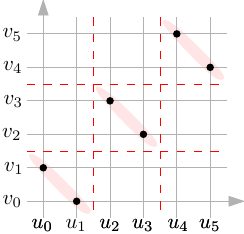}
		\caption{A 2-thick 3-twist and its grid representation, called a 2-thick increasing 3-sequence.}
		\label{fig:thick_twist}
	\end{subfigure}
	\begin{subfigure}[b]{\linewidth}
		\vspace{\baselineskip}
		\centering
		\includegraphics[page=7]{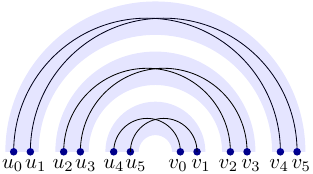}
		\hspace{4em}
		\includegraphics[page=8]{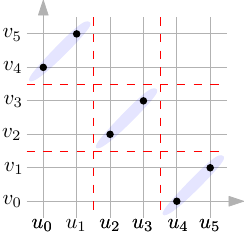}
		\caption{A 2-thick 3-rainbow and its grid representation, called a 2-thick decreasing 3-sequence.}
		\label{fig:thick_rainbow}
	\end{subfigure}
	\caption{Thick patterns (left) and their grid representations (right)}
	\label{fig:thick_pattern}
\end{figure}

Next we investigate the relation between thick patterns and $\diamondtimes$-patterns. Specifically, we show that thick patterns occur if and only if $\diamondtimes$-pattern occur, where the sizes are tied by a polynomial function (\cref{lm:diamond_thick}).
Together with \cref{thm:fixed_char} we obtain the following main result of the section.

\ThickPatternThm*

Note that the lower bound is already given by \cref{thm:fixed_char}, so our task is to bound the mixed page number in terms of the largest thick pattern.
\Cref{thm:fixed_char} also gives an upper bound using $ \diamondtimes $-patterns instead of thick patterns, so we aim to show that if there is a large $ \diamondtimes $-pattern in a matching, then it also contains a large thick pattern, which proves the upper bound of \cref{thm:thick_pattern}.
This is accomplished by identifying a thick pattern within a $\diamondtimes$-pattern with the help of grid representations.
To support the connection between text and figures, we use notation that matches the setting.
The edges of an (ordered) matching, $M$, are represented by \emph{points} in the grid representation.
Recall that a set of increasing points forms a twist in $ M $, whereas a set of decreasing points represents a rainbow.
Hence, when a set of edges forms a $t$-thick $k$-twist in the linear layout, the corresponding points are
called a \emph{$t$-thick increasing $k$-sequence}, and similarly a $t$-thick $k$-rainbow is called a \emph{$t$-thick decreasing $k$-sequence};
refer to \cref{fig:thick_twist,fig:thick_rainbow} for an illustration.
If any of the two is meant, we say \emph{$t$-thick $k$-sequence}, where $t$ is called the
\emph{thickness} and $ k $ the \emph{length} of the sequence.

\newpage
\begin{lemma}\label{lm:diamond_thick}
	Let $ M $ be a $ (2k)^\dexp $-$\diamondtimes$-pattern, then $M$ contains a $k$-thick pattern.
\end{lemma}
\begin{proof}
    The proof relies on the grid representation, which is subdivided repeatedly. First, after each subdivision, we find a sufficiently large $ \diamondtimes $-pattern, and second, these smaller patterns can be combined to a thick pattern.
    We assume that $ k = 2^t $ for some integer $t \geq 0 $ and pay for this by considering a $k^\dexp$-$\diamondtimes$-pattern instead of a $(2k)^\dexp$-$\diamondtimes$-pattern.
	So, consider the grid representation, $\Gamma$, of a $k^\dexp$-$\diamondtimes$-pattern and let $ n = k^\dexp $.
	In the first step, we partition the grid into four sub-squares of equal size, denoted by $Q_{1,1}$, $Q_{1,2}$, $Q_{2,1}$, and  $Q_{2,2}$ as shown in \cref{fig:grid_subdivision}.
	We claim that two of the four sub-squares in opposite corners contain at least $ n^2 / 4 $ points each.
    It is well known that such a subdivision always exists~\cite{BLP07}, but let us give an argument for this concrete one for the sake of completeness.
	Indeed, since there is exactly one point in every row of the grid, $Q_{1,1}$ and $Q_{1,2}$ together contain exactly $n^2/2$ points.
	It follows that either $Q_{1,1}$ or $Q_{1,2}$ contains at least $n^2/4$ points, say $ Q_{1, 1}$.
	Then, since there is exactly one point in every column, there are at most $n^2/4$ points in $Q_{2,1}$.
	Analogously, it follows that $Q_{2,2}$ contains at least $n^2/4$ points.
	By symmetry, we conclude that either both $Q_{1,1}$ and $Q_{2,2}$ or both $Q_{1,2}$ and $Q_{2,1}$ contain at least $n^2/4$ points each.

	Next we show how to find a $ \diamondtimes $-pattern of sufficient size whenever there are $ n^2/4 $ points in a sub-square.
	To this end, consider $ Q_{1,1} $ and assume it contains $ n^2/4 $ points; the other cases are symmetric.
    We show that if $Q_{1,1}$ does not contain a $q^*$-$\diamondtimes$-pattern for some carefully chosen value $q^*$, then it contains less than $n^2/4$ points, contradicting the assumption.
    For this, we consider the $n$ increasing sequences of the $n$-$\diamondtimes$-pattern and observe that an increasing sequence that has $q \geq 0$ points in $Q_{1,1}$ actually has the first $q$ points in the sub-square.
    In particular, if there are $ q $ increasing sequences having $ q $ points in $ Q_{1,1} $,
	then the sub-square contains a $ q $-$ \diamondtimes $-pattern.
	For a fixed $q$, we categorize the increasing sequences into
	\textit{long}, each containing at least $q$ points, and \textit{short}, each containing at most $q-1$ points.
	Now, we assume that $ Q_{1,1} $ contains no $ q $-$ \diamondtimes $-pattern and upper-bound the number of vertices in $Q_{1,1} $ depending on $q$.
    First, without a $q$-$\diamondtimes$-pattern, we have at most $q-1$ long sequences.
	Counting the maximum possible number of points that can be in the sub-square, we obtain at most $ q-1 $ long sequences
	contributing at most $ n $ elements plus the remaining short sequences each contributing at most $ q-1 $ points.
	This adds up to at most $ (q-1)n + (n - q + 1)(q-1) $, which equals $ n^2/4 $ for $ q = (1-\sqrt{3}/2)n + 1 $.
	That is, for every strictly smaller $ q^* $, say
	$ q^* = \lceil (1-\sqrt{3}/2)n \rceil \approx \lceil0.134 n\rceil$, avoiding a $ q^* $-$ \diamondtimes $-pattern contradicts the fact that $ Q_{1,1} $ contains at least $ n^2/4 $ points.
	Therefore, we find a $ q^* $-$ \diamondtimes $-pattern in $ Q_{1,1} $.

\begin{figure}
	\centering
	\includegraphics[page=5]{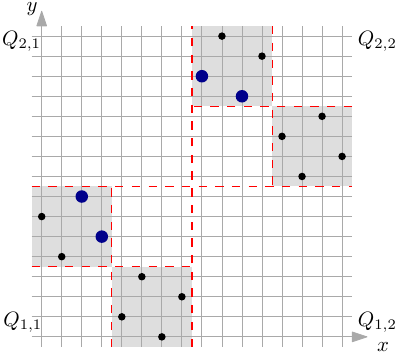}
	\caption{%
		A grid representation of a matching with a separated layout
		containing a $\diamondtimes$-pattern of size $4 \times 4 $ and a thick pattern of size $2 \times 2 $.
		The large sub-squares, $Q_{1,1}$, $Q_{1,2}$, $Q_{2,1}$, and $Q_{2,2}$, illustrate the first step of the subdivision.
        The four gray sub-squares, each containing a $2$-$\diamondtimes$-pattern, form the permutation $ \pi = (2, 1, 4, 3) $ in the final step.
        Taking a decreasing sequence of the first and the third sub-square then yields a $2$-thick increasing $2$-sequence, i.e., a $2$-thick twist (blue).
	}
	\label{fig:grid_subdivision}
\end{figure}

	Now, we take either $Q_{1,1}$ and $Q_{2,2}$ or $Q_{1,2}$ and $Q_{2,1}$, depending on where the most points are, and
	repeat the division into sub-sub-squares recursively.
	We apply the subdivision step $h = 2 \log_2 k $ times.
	In each recursion step, the side length of the largest $\diamondtimes$-pattern is multiplied by $ 1-\sqrt{3}/2 $, so
	the smallest sub-squares contain a $ \diamondtimes $-pattern with side length at least $n ( 1-\sqrt{3}/2 )^h$.
	With $k = 2^t$ for some $t \ge 0$ so that $h=2t$,
	the side length is lower-bounded as follows:
	\begin{equation}\label{eq:smallest_diamond}
		n ( 1-\sqrt{3}/2 )^h \ge k^7 (1 - \sqrt{3}/2)^h = 2^{7t} (1 - \sqrt{3}/2)^{2t}
		= (128 (1 - \sqrt{3}/2)^2)^t \ge 2^t = k.
	\end{equation}
	Therefore, each sub-square at the lowest level contains a $ k $-$ \diamondtimes $-pattern. In particular,
	every such sub-square contains an increasing \sqf-sequence (a $\sqf$-twist) and a decreasing \sqf-sequence (a $\sqf$-rainbow).

	The final step is to argue that there is a $k$-thick pattern in $\Gamma$.
	To this end, we build a permutation, $\pi$, formed by the sub-squares and illustrated in \cref{fig:grid_subdivision}. An element in the permutation
	is a sub-square with a $ k $-$ \diamondtimes $-pattern identified on the previous stage.
	The sub-squares are numbered with respect to their $x$-coordinates in $\Gamma$, while the
	order of the sub-squares in $\pi$ is determined by their $y$-coordinates. Note that since on every recursive step
	we consider the opposite sub-squares only, all the coordinates are distinct and $\pi$ has exactly $2^h$ elements.
	By the Erd\H{o}s-Szekeres theorem~\cite{ES35}, there exists an increasing subsequence in $\pi$
	of length at least $\sqrt{2^h} = k$ or a decreasing subsequence of that length.
	In the former case, the increasing subsequence of the sub-squares together with a decreasing $k$-sequence in each sub-square
	forms a $k$-thick increasing $k$-sequence (a $k$-thick twist); in the latter case, the decreasing subsequence forms a $k$-thick decreasing $k$-sequence (a $k$-thick rainbow).
\end{proof}

This proves the upper bound of \cref{thm:thick_pattern}, since a matching $ M $ without a $ k $-thick pattern has no $ (2k)^\dexp $-$ \diamondtimes$-pattern by \cref{lm:diamond_thick}, which implies $ \mn(M) \leq 2 \cdot (2k)^7 = 256 k^\dexp $ by \cref{thm:fixed_char}.
While the $\Oh(k^7)$ bound is most likely not best-possible, in \cref{fig:quadratic_lower_bound} we present a family of matchings with separated layouts
whose largest thick pattern is of size $ k \times k $, but the mixed page number is $ k^2 $ and the largest
$ \diamondtimes $-pattern is of size $ k^2 \times k^2 $.
This is in contrast to general $ \diamondtimes $-patterns, which guarantee a linear upper bound for the mixed page number.
We remark that the construction is also presented in~\cite{H23}.

\begin{figure}
    \centering 
    \includegraphics{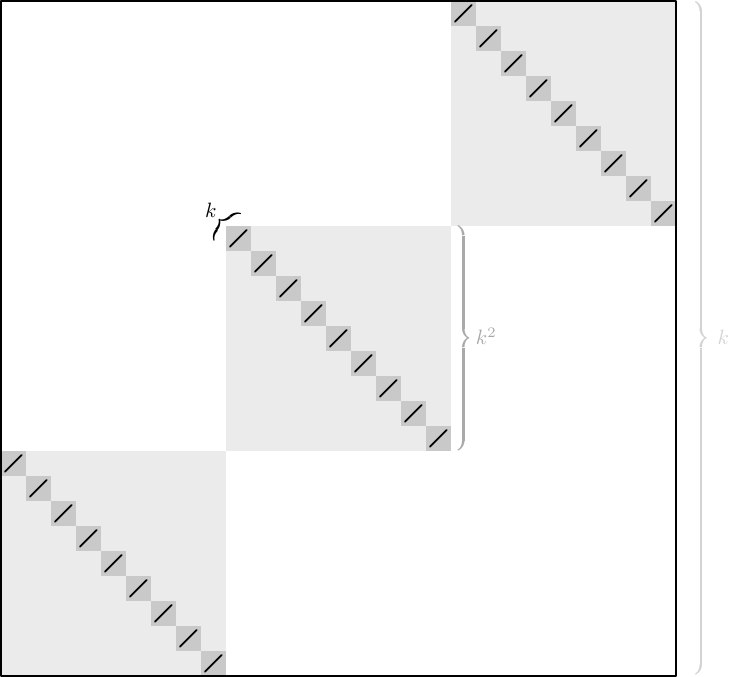}
    \caption{%
        Grid representation of a $\diamondtimes$-pattern of size $k^2 \times k^2$ and with mixed page number $k^2$ that does not contain a thick pattern of size larger than $k \times k$.
        It consists of $k$ big squares (light gray) arranged to the top right of each other, where each big square contains a $k$-thick decreasing $k^2$-sequence.
        Note that every thick decreasing sequence uses only points of one big square and thus has thickness at most $k$.
        On the other hand, a thick increasing sequence does not have two of its decreasing sequences in the same big square, which bounds its length by $k$.
        Hence, there is no $(k+1)$-thick pattern.
    }
    \label{fig:quadratic_lower_bound}
\end{figure}

\section{Non-Separated Matchings and Bounded-Degree Graphs}%
\label{sec:general_matching}

In this section we generalize our results to general (non-separated) layouts of ordered matchings and bounded-degree graphs.
For this let us first introduce the notion of quotients.
Consider a graph $ G $ with a fixed vertex order~$ \prec $ together with a partition $ \Ih $ of the vertices into \df{intervals}, i.e., such that the vertices in each part are consecutive.
The \df{quotient} $ G / \Ih $ of $ G $ and $ \Ih $ with respect to $ \prec $ is defined as the graph obtained from contracting each interval, together with the inherited vertex order.
The following lemma shows that mixed linear layouts of quotients can be transferred to the original matching.
This is then applied repeatedly to intervals, for which we start with a small interval whose mixed page number is easy to bound.
Step by step, we merge these intervals, considering more edges each time, whose mixed page number we bound using \cref{lem:quotient}.
This way we either can bound the mixed page number or find a large thick pattern, proving \cref{thm:general_matching}.

\begin{lemma}\label{lem:quotient}
	Let $ k \geq 1 $, let $ G $ be an ordered matching without a $ (k+1) $-thick pattern, let $ \Ih $ be a partition of the vertices into intervals, and let $ H = G / \Ih $ be the quotient.
	Then we have
	\[ \mn(G) \leq 6 \ell \cdot 2 \cdot (2k)^\dexp (1 + k + 14 (k+1) \log(k+1)) + m \in \Oh(\ell k^8 \log(k) + m), \]
	where $ \ell = \mn(H) $ is the mixed page number of the quotient, 
    and $m$ is the mixed page number of the sub-matching induced by all edges with both endpoints in the same interval of $ \Ih $.
\end{lemma}


Note that $m$ may also be defined as the maximum mixed page number induced by some interval, instead of the union of all such sub-matchings, which might be convenient for some applications.
However, this comes at the cost of using $ 2m $ pages for these edges, instead of only $ m $, since we need $ m $ stacks plus $ m $ queues, so that each interval finds the suitable number of each type, independent of how the pages are distributed to stacks or queues.
Since we apply \cref{lem:quotient} recursively, this makes an asymptotic difference.

We also remark that all arguments in the proof of \cref{lem:quotient} also work for general ordered graphs that are not necessarily a matching, except for two details:
The major problem for non-matchings is that we do not know how to handle the separated case.
That is, for non-matchings we obtain an upper bound of $ \Oh(\ell k \cdot f(k) + m) $, where $ f(k) $ is an upper bound for the separated case, replacing the $ 2 \cdot (2k)^7 $-bound from \cref{lm:diamond_thick,thm:fixed_char}.
Second, we slightly need to adapt how we choose edges to find a thick pattern, which we point out in a footnote during the proof.
However, we formulate the main body of the proof with matchings for ease of presentation.

\begin{proof}
    %
	We use $ m $ pages for the edges with both endpoints in the same interval, so our main task is to deal with edges in $ G $ between distinct intervals corresponding to edges of the quotient $ H $.
	We consider only a single page of $ H $, which we pay with a factor of $ \ell $.
	Now having a stack or a queue, it is well-known that it can be partitioned into six \df{one-sided} star forests, that is, for each part either all stars have their center to the left of their leaves or all to the right.%
	\footnote{
		Every subgraph $ H'' $ of a 1-stack graph or a 1-queue graph $ H' $ has mixed page number 1 and thus has at most $ 2 |H''| - 3 $ edges~\cite{DW04}.
        Hence, $ H' $ can be partitioned into two forests~\cite{nash-williams}, and rooting them gives an orientation of the edges such that each vertex has at most two outgoing edges, one per forest.
        For each vertex $ v $, we obtain a star with center $ v $ by taking all incoming edges.
        These stars cover all edges and can be partitioned into three star forests since $ H' $ is 3-colorable~\cite{DW04}.
        Now each star forest can be partitioned into two one-sided star forests.
	}
	At the cost of an additional factor of 6, we may consider a stack or a queue of $ H $ consisting of one-sided stars.
	Without loss of generality, we may assume that the stars have their center as their rightmost vertex.
	Note that a one-sided star in $ H $ induces a separated pattern in $ G $ consisting of a single interval (corresponding to the center of the star) on one side and possibly several intervals on the other.
	This is particularly convenient as we already know how to deal with separated matchings by \cref{thm:thick_pattern}, with an upper bound of $ 2 \cdot (2k)^7 $ given by \cref{lm:diamond_thick,thm:fixed_char}.
	We fix such a $ 2 \cdot (2k)^\dexp $-page assignment $ \alpha $ for each subgraph of $ G $ induced by some one-sided star.


	\begin{figure}
		\centering
		\includegraphics{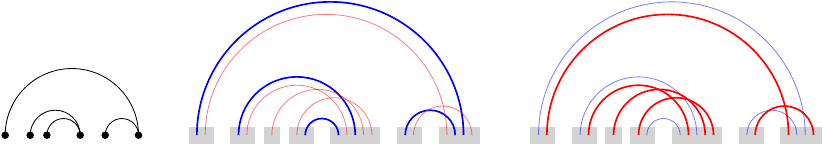}
		\caption{%
			A stack of the quotient $ H $ consisting of one-sided stars (left) and the matching induced in $ G $ (middle and right).
			Each star induces a separated pattern admitting a page assignment $ \alpha $ with at most $ 2 \cdot (2k)^\dexp $ stacks (blue) and queues (red).
			The set $ Q $ is formed by one queue per star (red).
			The stacks can be reused for distinct stars, while the edges in the queues of all stars together are partitioned into two groups -- one having only small twists and one having only small rainbows.
			The intervals of $ \Ih $ are indicated by gray boxes.
		}
		\label{fig:quotient_stack}
	\end{figure}

	\subparagraph*{Stacks of $ \bm{H} $.}
	We are now ready to construct a mixed linear layout for a subgraph of $ G $ induced by a one-sided star forest forming a stack in the layout of $ H $; see \cref{fig:quotient_stack}.
	First, edges in $ G $ that are assigned to a stack by $ \alpha $ are easy to handle:
	As the stars do not cross, we may reuse the same $ 2 \cdot (2k)^\dexp $ stacks.
	That is, we are left with the edges of $ G $ assigned to queues by $ \alpha $.
	Consider only one out of the up to $ 2 \cdot (2k)^\dexp $ queues of each star and let $ Q $ denote the set of these edges.
	We pay this with a factor of $ 2 \cdot (2k)^\dexp $.
	As the stars may nest, we cannot simply join the queues but need a more involved strategy, refer to \cref{fig:quotient_Q_L} for an overview of the strategy.
	Observe that the edges of $ Q $ belonging to the same star in $ H $ are separated but do not nest, that is, they form a twist.
	Next we aim to use these twists to either cover the edges with few pages or to find a large thick rainbow.
	For each twist, take the $ k $ leftmost edges%
    \footnote{%
        If $ G $ is not a matching, then the edges of $ Q $ belonging to the same star in $ H $ are not necessarily a twist, since they can either cross or share an endpoint (but do not nest). In this case, instead of taking the $ k $ leftmost edges, we take edges from left to right (by first endpoint, secondary by second endpoint) as long as there is no $ (k+1) $-twist. 
    }
    (or all if there are less than $ k $), and denote the union of all these edges of all twists by $ L $.
	As the stars do not cross, 
    $ k $ stacks suffice for $ L $.

    \begin{figure}
        \centering
        \includegraphics{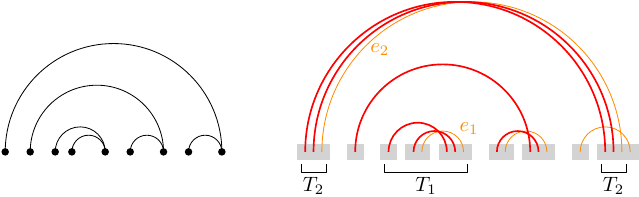}
        \caption{%
            A stack of the quotient $H$ consisting of one-sided stars (left) and the matching induced in $G$ restricted to one queue per star, named $Q$ (right).
            Note that for each star, the corresponding edges in $Q$ form a twist and that twists corresponding to different stars do not cross.
            The set $L$ consists of the (up to) $k$ leftmost edges of each twist and are highlighted thick for $k=2$.
            Observe that if two edges in $Q-L$ nest (e.g., $e_1$ and $e_2$), this yields two $(k+1)$-twists, one nesting above the other ($T_1$ and $T_2$). 
            See \cref{fig:quotient_rainbow} for details.
        }
        \label{fig:quotient_Q_L}
    \end{figure}

	\begin{figure}
		\centering
		\includegraphics{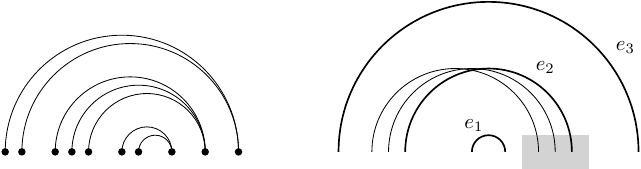}
		\caption{%
			Left: Non-crossing one-sided stars in $ H $.
			Right: A rainbow $ e_1, e_2, e_3 \in Q - L \subseteq E(G) $ together with the twist $ T_2 $.
			Recall that $ e_2 $ is chosen such that it is the rightmost edge of $ T_2 $.
			The right endpoints of $ T_2 $ are consecutive as they are in a common interval that is contracted to the center of the star in $ H $.
			The left endpoints are consecutive as the respective edges in the stars do not cross and thus the edges of $ T_2 $ do not cross any edges of other stars like $ e_3 $.
		}
		\label{fig:quotient_rainbow}
	\end{figure}

	For the remaining edges, we show that a rainbow implies a thick rainbow of the same size in $ G $, and thus they can be covered by $ k $ queues by assumption.
	Indeed, consider a rainbow $ e_1, e_2, \dots $ in $ Q - L $, where the edges nest above each other in this order.
	We refer to \cref{fig:quotient_rainbow} for an illustration of the upcoming argument.
	Recall that the edges of $ Q $ belonging to the same star form a twist, and in particular do not nest, and thus the edges in the rainbow belong to pairwise distinct stars (\cref{fig:quotient_Q_L}).
	Let $ T_i $ denote the $ (k+1) $-twist consisting of $ e_i $ and the $ k $-twist in $ L $ induced by the same star.
	Also recall that the stars are in a stack of $ H $ and as their edges do not cross, the same holds for edges in $ G $ belonging to distinct stars.
	In particular, the edges of the twist $ T_i $ do not cross $ e_{i + 1 } $ nor $ e_{i-1} $.
	Therefore, the left endpoints of $ T_i $ are between the left endpoints of $ e_i $ and $ e_{i + 1} $.
	For the right endpoints of $ T_i $, recall that the stars are one-sided with the center to the very right.
	As such, the right endpoints of $ T_i $ belong to a common interval of $ \Ih $, and thus are to the right of $e_{i-1}$.
	Since $ T_i $ is a twist, its right endpoints are between the right endpoints of $ e_{i-1} $ and $ e_i $.
	This yields $ (k+1) $-twists that are pairwise nesting, that is, if there is a $ (k+1) $-rainbow in $ Q - L $, then we obtain a $ (k+1) $-thick rainbow in $ G $, a contradiction.
	Hence, there is no $ (k+1) $-rainbow in $ Q - L $ and $ k $ queues suffice for these edges.

	To sum up, we have $ 2 \cdot (2k)^\dexp $ stacks for the edges assigned to stacks by $ \alpha $, plus $ k $ stacks for $ L $ and $ k $ queues for $ Q - L $, which we do $ 2 \cdot (2k)^\dexp $ times as this is the number of queues $ \alpha $ may use.
	This yields $ 2 \cdot (2k)^\dexp + (k + k) \cdot  2 \cdot (2k)^\dexp = 2 \cdot (2k)^\dexp \cdot (2k + 1) $ for all edges induced by one star forest in $ H $ that is in a stack in the layout of $ H $.


    \begin{figure}
        \centering
        \includegraphics{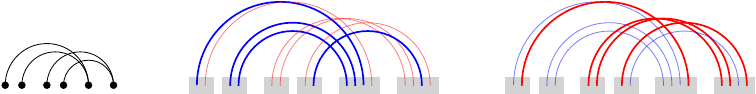}
        \caption{%
            Analogous situation of \cref{fig:quotient_stack} with a queue of $H$ (left) instead of a stack. 
            Edges of $G$ assigned to a queue by $\alpha$ (red) may all share the same queue.
            The other edges, assigned to a stack by $\alpha$ (blue), form a rainbow for each star of $H$. 
        }
        \label{fig:quotient_queue}
    \end{figure}

	\subparagraph*{Queues of $\bm{H}$.}
	The approach for a queue of $ H $ is fairly similar, compare \cref{fig:quotient_stack,fig:quotient_queue}.
	As discussed above, we have one-sided stars in the queue and a fixed page assignment $ \alpha $ for the separated subpatterns induced by each star individually.
	All edges in $ G $ that are assigned to a queue by $ \alpha $ may share the same $ 2 \cdot (2k)^\dexp $ queues so this time the edges assigned to stacks are the challenge.
	Take one stack per star and let $ S $ be the union of these edges.
	Again, we pay this with a factor of $ 2 \cdot (2k)^\dexp $ as $ \alpha $ uses up to $ 2 \cdot (2k)^\dexp $ stacks.
	First, consider a single star and the edges of $ S $ it induces in $ G $.
	As these edges are separated and fit into a stack, they form a rainbow.
	That is, $ S $ is the disjoint union of rainbows, one per star.
	Now, for each rainbow, take the $ k $ bottom-most edges%
    \footnote{%
        If $ G $ is not a matching, then the edges of $ S $ belonging to the same star in $ H $ do not necessarily form a rainbow, since they can either nest or share an endpoint (but do not cross). In this case, instead of taking the $ k $ bottom-most edges, we take edges from bottom to top as long as there is no $ (k+1) $-rainbow. 
    }
    (or all if there are less than $ k $) and let $ B $ denote the union of all these edges.
	As the edges of the stars are in a queue of $H$ and thus do not nest, the same holds for the edges of distinct stars.
	Hence, the largest rainbow in $ B $ has size at most $ k $ and $ B $ can be covered by $ k $ queues.

	The main difference between stacks and queues of $ H $ is how we handle the remaining edges in $ S - B $.
	We aim to show that a twist in $ S - B $ implies a thick twist in $ G $ of similar (but not exactly the same) size.
	So consider a twist in $ S - B $ with edges $ e_1, e_2, \dots $ from left to right.
	For each $ e_i $, let $ R_i $ denote the $ (k+1) $-rainbow consisting of $ e_i $ and the $ k $-rainbow in $ B $ belonging to the same star.
	Recall that the stars have their centers to the right, that is, for each star the right endpoints of the edges induced in $ G $ are consecutive.
	To find a thick twist, we also need the left endpoints to be consecutive.
	Indeed, this is the case for all but the rightmost rainbow, as otherwise some edge of $ R_i $ nests below $ e_{i+1} $; see \cref{fig:quotient_twist}.
	This however, would also imply that the respective edges in $ H $ nest, which is not the case as we consider a queue of $ H $.
	We conclude that if there is a $ (k+2) $-twist in $ S - B $, then we obtain a $ (k+1) $-thick twist in $ G $, a contradiction.
	Therefore, the edges in $ S - B $ can be covered by $ 14 (k + 1) \log(k+1) $ stacks~\cite{Dav22}.

	\begin{figure}
		\centering
		\includegraphics{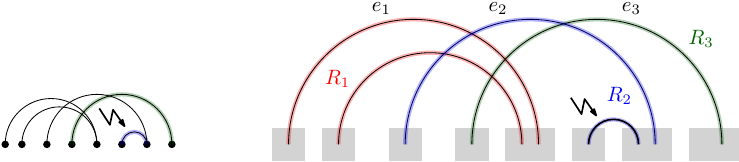}
		\caption{%
			Rainbows $ R_1, R_2, R_3 $ with their topmost edges $ e_1, e_2, e_3 $ forming a twist (right).
			If an edge of $ R_2 $ does not have its left endpoint close to $ e_2 $, then it nests below $ e_3 $.
            This is a contradiction since it yields two nesting edges in a queue of $H$ (left).
		}
		\label{fig:quotient_twist}
	\end{figure}

	Adding up all pages, we have $ 2 \cdot (2k)^\dexp $ queues for the edges assigned to queues by~$ \alpha $, plus $ 2 \cdot (2k)^\dexp $ times $ k $ queues for $ S $ and $ 14 (k + 1) \log(k+1) $ stacks for $ S - B $, giving $ 2 \cdot (2k)^\dexp \cdot (1 + k + 14 (k + 1) \log(k+1)) $.
\end{proof}

Next we use \cref{lem:quotient} to prove our main result of this section, namely that it suffices to bound the size of the largest thick pattern to obtain bounded mixed page number.
Recall that this is also necessary as $ k $-thick patterns have mixed page number $ k $.

\GeneralMatchingThm*


The lower bound is already clear, so we focus on the upper bound.
For ease of presentation, we rephrase the upper bound with an index shift, which is absorbed by the big-O notation.

\begin{lemma}\label{lem:general_matching}
    If an ordered matching $M$ contains no $k$-thick pattern, then the mixed page number of $M$ is at most $ \Oh(k^{10} \log^2 k) $.
\end{lemma}
We remark that none of our arguments needs that $ M $ is a matching, except for \cref{lem:quotient}, which relies on a bound on the mixed page number in the separated case, which we leave open for large-degree graphs.
That is, the following proof reduces the problem of characterizing ordered graphs via thick patterns to the separated setting.

\begin{figure}
	\centering
	\includegraphics{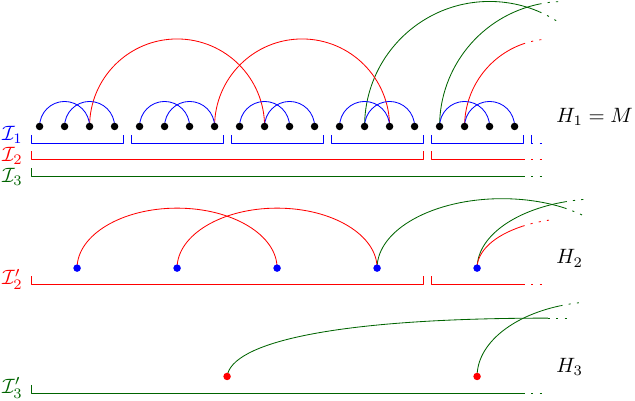}
	\caption{%
        Intervals for \cref{lem:general_matching}. 
        Top: In the first step, \cref{lem:quotient} is applied to $ M_2 $ (blue and red edges) and $ \Ih_1 $ (blue intervals). In the second step, \cref{lem:quotient} is applied to $ M_3 $ (all shown edges) and $ \Ih_2 $ (red intervals). 
        Middle: $ H_2 = M / \Ih_1 $ with intervals $ \Ih_2' $ corresponding to $ \Ih_2 $. The first application of \cref{lem:quotient} uses that the quotient $ M_2 / \Ih_1 $ (red edges of $ H_2 $) has small stack number by construction.
        Bottom: $ H_3 = H_2 / \Ih_2' = M / \Ih_2 $ with intervals $ \Ih_3' $.
    }
	\label{fig:interval_hierarchy}
\end{figure}

\begin{proof}
    Our goal is to find intervals that allow us to apply \cref{lem:quotient} to $ M $.
    For this, we need a bound on the mixed page number of the quotient and of the sub-matching induced by all edges with both endpoints in the same interval, where the latter will be the hard part. 
	The idea is to start with small intervals, repeatedly merge them into larger ones, and in each step apply \cref{lem:quotient} to bound the mixed page number within the newly created intervals.
	More detailed (see also \cref{fig:interval_hierarchy}), we start with a partition $ \Ih_1 $ into intervals having bounded stack number, which are merged into a partition $ \Ih_2 $ of larger intervals.
    We aim to do this in such a way that for the sub-matching $ M_2 $ consisting of all edges with both endpoints in the same interval of $ \Ih_2 $, the quotient $ M_2 / \Ih_1 $ has small mixed page number.
	Applying \cref{lem:quotient} yields bounded mixed page number for $ M_2 $.
	Repeating this, we obtain a hierarchy of intervals until all vertices are within a single interval and \cref{lem:quotient} is applied to $ M $.
	Crucially, we show that this happens after at most $ k $ steps, otherwise we find a $ k $-thick pattern.
	
	To define the partitions $ \Ih_1, \Ih_2, \dots $ of the vertex set $ V(M) $ of $ M $, we first define a sequence of ordered graphs $ H_1, H_2, \dots $, together with a partition $ \Ih_i' $ of the vertex set $ V(H_i) $ of $ H_i $ into intervals.
	The graphs $ H_i $ are quotients whose vertices we partition into intervals, giving the next quotient.
	See again \cref{fig:interval_hierarchy} for an example.
	To start with, let $ H_1 = M $.
	For $ i \geq 1 $, let $ \Ih_i' $ be the partition of $ V(H_i) $ into intervals $I^i_1, I^i_2, \dots$ defined as follows.
	In case $ H_i $ does not contain a $ (k+1) $-twist, $ \Ih_i' $ consists of a single interval containing all vertices of $ H_i $ and we stop the procedure.
	Otherwise, the first interval $I^i_1$ starts with the first vertex of $H_i$ in the given vertex order. 
	From left to right, we add vertices to $I^i_1$ until it contains a $(k+1)$-twist. 
	In particular, the last vertex of $I^i_1$ is the rightmost right endpoint of a $(k+1)$-twist and is followed by the first vertex of $ I^i_2 $.
	The next interval $I^i_2$ is then defined in the same way, that is, starting from its first vertex, it includes the minimum number of vertices such that it contains a $ (k+1) $-twist (or all remaining vertices if there is no $ (k+1)$-twist left).
	We continue until every vertex is in one of the intervals. 
	Then, $H_{i+1}$ is defined as the quotient $H_i / \Ih_i'$.
	Note that every interval, except possibly the last one, contains a $(k+1)$-twist.
	On the other hand, by construction, there is no $(k+2)$-twist in the subgraph of $H_i$ induced by any of the intervals in $ \Ih_i' $, so the mixed page number within the intervals is bounded by $\mathcal{O}(k \log k)$~\cite{Dav22}.
	Finally, the desired partition $ \Ih_i $ of $ V(M) $ is obtained by replacing the vertices of $ H_i $ in $ \Ih_i' $ by the corresponding vertices of $ M $.
	In particular, we have $ M / \Ih_i = H_i / \Ih_i' = H_{i+1} $.

	Next we show that the procedure stops with $ H_k $ (or before).
	Suppose to the contrary that there is always a $(k+1)$-twist in $H_i$ for $i = 1,\dots,k$.
	We show that this implies a $k$-thick rainbow in $M$, which is a contradiction.
	That is, we now look for $k$-twists, one in each $ H_i $, that nest above each other.
	For this, consider a $(k+1)$-twist with vertices $ u_1, \dots, u_{k+1}, v_1, \dots, v_{k+1} $ in $H_i$, for $2 \leq i \leq k$, and note that the first $k$ edges form a $k$-twist which nests above $ u_{k+1}$.
	Now recall that every vertex of $ H_i $ (except for the last) is the result of contracting an interval of $ H_{i-1} $ containing a $ (k+1) $-twist.
	Thus, a $(k+1)$-twist in $H_i$ corresponds to a $k$-twist nesting above a $(k+1)$-twist of $H_{i-1}$.
	By transitivity of the nesting relation, this gives a $k$-thick rainbow, a contradiction.
    Hence, $ H_k $ does not contain a $ (k+1) $-twist.

	By construction, it follows that $ \Ih_k' $ consists of a single interval containing all vertices of $ H_k $, and thus $ \Ih_k $ also consists of a single interval, containing all vertices of $ M $.
	That is, we aim to show that the mixed page number of the matching $ M_k = M $ induced by all edges within the unique interval of $ \Ih_k $ is at most $ \Oh(k^{10} \log^2 k) $.
	We do so by inductively showing that for each $ i \in [k] $ 
    the matching $ M_i $ induced by all edges of $ M $ with both endpoints in the same interval of $ \Ih_i $
    has mixed page number $ i \cdot c^2 \cdot k^9 \log^2 k $, where $ c \leq 57600 $ is the leading constant from the big-O notation in \cref{lem:quotient}. 
	For $ i = 1 $, by construction $ H_1 = M $ does not contain a $ (k+2) $-twist in any interval of $ \Ih_1' = \Ih_1 $, hence we have $ \mn(M_1) \leq \sn(M_1) \leq 14 (k+1) \log (k+1) $ by Davies~\cite{Dav22}.
    Without loss of generality $ c $ is sufficiently large such that this is at most $ c \cdot k \log k $.
	For $ i > 1 $, note that while the intervals in $ \Ih_i' $ do not contain a $ (k+2) $-twist in the quotient $ H_i $, the intervals in $ \Ih_i $ may contain arbitrarily large twists in $M$ (recall \cref{fig:quotient_stack,fig:quotient_Q_L,fig:quotient_rainbow,fig:quotient_queue} for matchings containing large twists while their quotients do not).
	However, \cref{lem:quotient} already deals with this situation.
    %
    That is, we next apply \cref{lem:quotient} to $ M_i $ and $ \Ih_{i-1} $ to find a page assignment for the edges in $ M_i $ that are not already covered by $ M_{i-1} $.
    Note that these are the edges in $ M_i $ whose endpoints are in distinct intervals of $ \Ih_{i-1} $.
	To do so, we need three properties required by \cref{lem:quotient}:
    \begin{itemize}
	    \item First, $ M_i $ needs to be a matching without a $ k $-thick pattern. This is inherited from~$ M $.
	    \item Second, we need an upper bound on the mixed page number of the quotient $ M_i / \Ih_{i-1} $.
        For this, recall that each interval in $ \Ih_i $ is constructed as the union of intervals in $ \Ih_{i-1} $ so that the quotient $  M / \Ih_{i-1} = H_i $ does not contain a $ (k+2) $-twist within any interval of~$ \Ih_i' $.
        Since $ M_i $ contains only edges that are entirely within some interval of $ \Ih_i $, it follows that $ M_i / \Ih_{i-1} \subseteq H_i $ does not have a $ (k+2) $-twist and thus has stack number at most $ \ell = c \cdot k \log k $ by Davies~\cite{Dav22}.
        %
	    \item The third requirement of \cref{lem:quotient} is an upper bound on the mixed page number of the sub-matching $ M_{i-1} $ of $ M $ consisting of all edges with both endpoints in some interval of $ \Ih_{i-1} $, for which we have an upper bound of $ m = (i-1) \cdot c^2 \cdot k^9 \log^2 k $ by induction.
	\end{itemize}
	Having all prerequisites, we now may apply \cref{lem:quotient} to $ M_i $ and $ \Ih_{i-1} $, with which we obtain
	$ \mn(M_i) 
		\leq c \cdot \ell k^8 \log k + m 
		= c \cdot c k \log k \cdot k^8 \log k + (i-1) \cdot c^2 \cdot k^9 \log^2 k
		\leq i \cdot c^2 \cdot k^9 \log^2 k.
	$
    This finishes the induction, and since $ M_k = M $, we take $ i = k $ to get our desired upper bound of $ \mn(M) \leq k \cdot c^2 \cdot k^9 \log^2 k = c^2 \cdot k^{10} \log^2 k $.
\end{proof}

With Vizings's Theorem~\cite{Viz65}, \cref{thm:general_matching} generalizes to bounded-degree graphs.
To this end, we find a proper edge-coloring of a graph and apply \cref{thm:general_matching} to each of the matchings.
Note that the size of the largest thick pattern gives a trivial linear lower bound.

\BoundedDegreeGraphThm*

\section{Complexity Results}%
\label{sec:complexity}

Having understood mixed linear layouts of ordered matchings via $\diamondtimes$-patterns and thick patterns, let us briefly discuss that the latter two are not only structurally easier but also computationally.
While it is known that, for every $s \geq 2, q \geq 1 $ or $s \geq 4, q \geq 0$, it is \NP-hard to decide whether a given ordered graph admits an $s$-stack $q$-queue layout~\cite{CKN19,Un88}, to the best of our knowledge, the complexity of determining the mixed page number is not known.
Note that the difference between the two problems is that the mixed page number only bounds the sum of stacks plus queues but how the sum is split is left to the problem.
To compare with, we first show that it is \NP-hard to decide whether the mixed page number is at most~$k$, even for fixed $k \geq 4$ and ordered matchings.

\begin{theorem}\label{thm:mn_hardness}
    For every fixed $k \geq 4$, it is \NP-complete to decide whether a given ordered matching has mixed page number at most $k$.
\end{theorem}

\begin{proof}
    Unger~\cite{Un88} shows that $k$-coloring circle graphs is \NP-hard for every $k \geq 4$.
    Here, the input is a set of chords of a circle, where no two chords share an endpoint, and the task is to color them with $k$ colors such that no two chords of the same color cross.
    This is equivalent to finding a stack layout for an ordered matching, which is the problem we reduce from.
    So, given an ordered matching $M$, we append a $k$-thick $(k+1)$-rainbow to the very right, so that the new edges neither cross nor nest with edges from $M$.
    Let $M'$ denote the resulting ordered matching.
    
    We claim that $\sn(M) \leq k$ if and only if $\mn(M') \leq k$, which proves the \NP-hardness.
    First, a $k$-thick $(k+1)$-rainbow consists of $k$-twists that are nesting, so $k$ stacks can be reused to cover all $k$-twists.
    Thus, if $k$ stacks suffices for $M$, the same holds for $M'$, and we have $\mn(M') \leq \sn(M') \leq k$.
    For the other direction, we show that only stacks are used for the $k$-thick $(k+1)$-rainbow.
    Observe that a queue can cover at most $k$ edges of a $k$-thick $(k+1)$-rainbow, namely one of the twists, since any two edges from distinct twists nest.
    Similarly, a stack can cover at most $k+1$ edges, one of each twist.
    Now to cover all $k(k+1)$ edges with $k$ pages, each page needs to contain $k+1$ edges and thus is a stack.
    Therefore, if $\mn(M') \leq k$, then we have $ \sn(M') \leq k$, which upper-bounds the stack number of $M$.
\end{proof}

In contrast to \cref{thm:mn_hardness}, $k$-$\diamondtimes$-patterns can be found efficiently for every fixed $k$ and thick patterns even for arbitrary size.
Additionally, recall from \cref{sec:diamond} that $\diamondtimes$-patterns allow for a 2-approximation of the mixed page number of matchings with separated layouts, but this does not give the largest $\diamondtimes$-pattern (recall \cref{fig:square_diamond_reverse}).

\begin{theorem}\label{thm:diamond_fpt}
    Finding a largest $ \diamondtimes $-pattern in an ordered matching is FPT in its size.
    Moreover, checking for a fixed $k$-$\diamondtimes$-pattern in an ordered matching on $m$ edges can be done in time $2^{\Oh(k^4 \log k)} \cdot m^2$.
\end{theorem}

\begin{proof}
    The number of $k$-$\diamondtimes$-patterns is bounded by a function in $k$, so it suffices to consider a single pattern to obtain FPT time.
    We use that $ \diamondtimes $-patterns are separated, hence we only need to search for them in separated sub-matchings.
    For this, fix an ordered matching with vertices $ w_1, \dots, w_{2m} $ from left to right and let the \emph{$ w_i $-$w_{i+1} $-gap matching} be the separated sub-matching that consists of all edges with one endpoint at $ w_i $ or to the left, and the other endpoint at $ w_{i+1} $ or to the right.
    Note that every separated sub-matching is contained in some $ w_i $-$w_{i+1} $-gap matching, so we only need to search within those.


    To deal with the separated case, observe that separated matchings correspond bijectively to permutations by labeling the edges by $1, \dots, m$ according the vertex order of their left endpoints and then obtaining the permutation from the order of their right endpoints.
    Symmetrically, a permutation corresponds to the ordered matching obtained by permuting the right endpoints accordingly, see \cref{fig:permutations}.
    Thus, the task is to find a permutation of size $k \times k$ in a permutation of size $m$, which is feasible in time $2^{\Oh(k^4 \log k)} \cdot m$ by Guillemot and Marx~\cite{GM14}. 
    Doing this for each of the $ 2m - 1 $ $ w_i $-$w_{i+1} $-gap matchings yields the desired runtime.
\end{proof}

Finally, the recognition of thick patterns in a given ordered matching is polynomial, both in the size of the ordered matching and in the size of the thick pattern.
We remark that it is crucial here that, in contrast to $\diamondtimes$-patterns, the number of thick patterns of size $k \times k $ does not depend on $k$.

\begin{theorem}\label{thm:thick_poly}
   Given an ordered matching on $m$ edges, the size of a largest thick rainbow, thick twist, or thick pattern can be determined in time $\Oh(m^6 \log(m))$.
\end{theorem}

\begin{figure}
    \centering
    \includegraphics{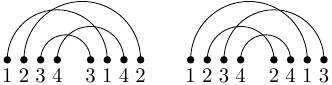}
    \caption{%
        Matchings with separated layout corresponding to the permutations $ 3142 $ (left), respectively $ 2413 $ (right).
        Observe that neither of these two ordered matchings is containing in a thick pattern.
    }
    \label{fig:permutations}
\end{figure}

\begin{proof}
    Ibarra~\cite{Iba97} shows that it can be decided in time $\Oh(m^5)$ whether a permutation~$\pi$ of length $m$ contains a permutation $\tau$, provided that $\tau$ contains neither $ 3142 $ nor $ 2413 $ as subpermutation.
    To make use of this, we again exploit that matchings with separated layouts correspond bijectively to permutations via the order of their right endpoints with respect to their left endpoints, recall \cref{fig:permutations}.
    That is, a $k$-thick twist corresponds to the permutation 
    $ k, \dots, 1, 2k, \dots, k+1, \dots, k^2, \dots, (k-1)k + 1 $
    and a $k$-thick rainbow corresponds to 
    $ (k-1)k + 1, \dots, k^2, \dots, k+1, \dots, 2k, 1, \dots, k $.
    Observe that all maximal monotonically decreasing, respectively increasing, subsequences appear consecutively, and thus both sequences do not contain the two forbidden subsequences $ 3142 $ and $ 2413 $.

    That is, for a matching on $m$ edges with a separated layout, it can be decided in time $ \Oh(m^5) $ whether it contains a fixed thick pattern of arbitrary size.
    Since there are only linearly many possible sizes for a thick pattern, the maximum size can be binary-searched, giving an additional factor of $\log(m)$.
    Finally, it again suffices to search for a thick pattern in the $2m-1$ separated sub-matchings we obtain by separating the linear order between any two consecutive vertices and only considering the edges that go from one side to the other.
    Together, this yields a runtime of $ \Oh(m^5 \log(m) \cdot (2m-1)) = \Oh(m^6 \log(m)) $.
\end{proof}


\section{Critical Graphs}\label{sec:critical}

In this section, we focus on the exact characterization of
mixed linear layouts. 
In contrast to previous sections, it now does not suffice to bound the mixed page number by some function in the size of a forbidden pattern, but we are interested in the exact mixed page number.
That is, the goal is to construct, for an integer $k>0$,
a finite obstruction set of patterns for ordered graphs with mixed page number exactly $k$.
To this end, we introduce the concept of \df{critical graphs} that are minimal graphs that do not
admit a mixed linear layout with a certain number of pages.
Suppose $G=(V, E)$ is a graph
with a fixed vertex order $\prec$. For $s,q \ge 0$, we call $G$ \df{$(s,q)$-critical} if it does
not admit an $s$-stack $q$-queue layout under $\prec$ but every subgraph $G-e$ (that is, $G$ after deleting edge $e$)
for all edges $e \in E$ does. Similarly, \df{$k$-critical graphs} are minimal graphs not admitting a layout on
mixed $k = s + q$ (for some $s, q \ge 0$) pages.

We investigate critical graphs both for separated and general (non-separated) layouts.
As mentioned in \cref{sec:related_work}, separated matchings can be viewed as permutations;
prior results for the setting imply that the number of critical separated matchings is finite~\cite{KSW96,FH06,Wa21}.
We extend the results to separated bounded-degree graphs and to separated $2$-critical graphs,
thus, showing the existence of a characterization of separated layouts with a fixed mixed page number
for the classes of graphs.
For the non-separated case, however, we construct an infinite set of $k$-critical matchings
for every $k \ge 2$ and an infinite set of $(s, q)$-critical matchings for every $s \ge 2$, $q \ge 0$.
The construction answers the exact version of \cref{op:mn_characterization} in negative.
Our results of this section are summarized in \cref{fig:critical_graphs_overview}.

\begin{figure}
	\includegraphics{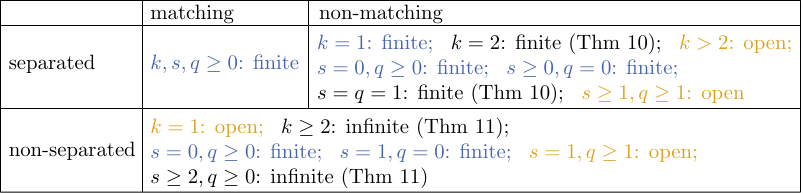}
	\caption{An overview visualizing whether the number of $k$-critical and $(s,q)$-critical graphs is finite or infinite in the cases
	of separated/non-separated layouts of matchings/non-matchings. The blue cases follow immediately from previous results,
    which we discuss in \cref{sec:critical_separated} and \cref{sec:critical_general}.}
	\label{fig:critical_graphs_overview}
\end{figure}

\subsection{Critical Graphs for Separated Layouts}
\label{sec:critical_separated}

A separated matching $G$ corresponds to a permutation, $\pi(G)$, of the right endpoints of the edges
relative to the left endpoints. Hence, the separated mixed page number of a
matching $G$ equals the minimum number of monotone subsequences that $\pi(G)$ can be partitioned
into. It is known that the set of permutations that can be partitioned into a fixed
number of monotone subsequences are characterized by a finite set of forbidden
subsequences~\cite{KSW96,FH06,Wa21}. Thus, for every pair $(s,q)$ with $s,q \ge 0$, there exists
a finite number of $(s,q)$-critical matchings with a separated layout.
Our main result in the subsection is the following generalization to non-matching graphs.

\CriticalSeparatedThm*

Before proving \cref{thm:critical_sep}, we remark that a similar characterization of graphs with a
bounded separated pure stack (queue) number is straightforward. In fact, there exists a single separated
$(s, 0)$-critical graph and a single separated $(0, q)$-critical graph for all $s, q \ge 1$, namely the $ (s+1) $-twist and the $(q+1)$-rainbow, respectively. 
In contrast, characterizations for mixed linear layouts are significantly more complex. 
We computationally
identified all $9$ graphs that are $1$-critical, all $20$ graphs that are $(1, 1)$-critical, and
all $3128$ graphs that are $2$-critical (in the separated setting).
To do so, the main challenge is to argue that the maximum degree is bounded.

Our proof of \cref{thm:critical_sep} uses similar arguments as~\cite{Wa21} who show, in our language,
that the number of edges in critical separated matchings is bounded.
The following lemma, provided by W\"arn \cite{Wa21}, is the combinatorial key to proving \cref{thm:critical_sep}.
It shows that partial Boolean functions, each representing an $s$-stack $q$-queue layout of $G-e$ for an edge $e$, can be
combined into a total Boolean function, representing an $s$-stack $q$-queue layout of $G$, provided that the number of edges is
sufficiently large and that the partial functions do not differ too much.
Moreover, the combined Boolean function has the property that it coincides with one of the partial Boolean functions for every subset of the domain, respectively of the edges, that is not too large.
We say that two Boolean functions $ P \colon D \to \{0,1\}, P' \colon D' \to \{0,1\}$ are $d$-close if they differ in at most $d$ elements, i.e., if $ P(x) \neq P'(x) $ for at most $d$ elements $x \in D \cap D'$.

\begin{lemma}[{\cite[Lemma 13]{Wa21}}]\label{lm:Boolean_functions}
	Let $r,d \in \mathbb{N}$, and let $U$ be a non-empty set.
	For each $i \in U$, suppose there is a Boolean function $P_i: U \setminus \{i\} \rightarrow \{0,1\}$ such that the $P_i$ are pairwise $d$-close.

	Then, there exists an integer $N = N(r,d)$ such that
    if $|U| > N$, there is a Boolean function
	$P: U \rightarrow \{0,1\}$ satisfying the following condition:
	For every $R \subseteq U$ of size at most $r+1$, there exists an $i \in U \setminus R$
	such that $P(k) = P_i(k)$ for every $k \in R$.
\end{lemma}

We remark W\"arn~\cite{Wa21} requires $ r \geq 2$ and provides that $ N(r,d) \leq \min(4r^{d+1},(4r)^{d/2+1}) $ in this case.
However, note that \cref{lm:Boolean_functions} also holds for $r = 1$, aside from the upper bound on $N(r,d)$.
To this end, observe that the only part where $r$ plays a role is the requirement that the last condition
holds for all subsets $R$ of size \textit{at most} $r+1$. In particular, when $r=1$, the same Boolean
function $P$ as for $r = 2$ can be used.
So let us define $N(1,d) = N(2,d)$ for ease of presentation.

Before using \cref{lm:Boolean_functions} to prove the following lemma, let us describe more detailed how it translates to mixed linear layouts.
For this, consider an ordered graph $G$ and let $U$ be the set of edges.
The Boolean functions correspond to mixed linear layouts, where an edge is in a queue if and only if it is assigned to 0, and in a stack otherwise.
That is, the $P_i$ are mixed linear layouts of $G - e $ for some edge $e$, which fits particularly well to the concept of critical graphs.
Further, two page assignments of subgraphs of $G$ are $d$-close if there are at most $d$ edges at which they disagree whether to assign the edge to a queue or to a stack.
Now, \cref{lm:Boolean_functions} requires mixed linear layouts for each subgraph $ G-e$ for some edge $e$ that are pairwise $d$-close and that the number of edges of $G$ is sufficiently large, depending on $r$ and $d$.
From this, the lemma builds a mixed linear layout for $G$ such that for every set $R$ on $r + 1$ edges, the layout of $G$ agrees with one of the provided partial layouts.
We remark that the interpretation of \cref{lm:Boolean_functions} does not require the layouts to be separated or have bounded degree, but both properties are exploited to prove closeness.

\begin{lemma}\label{lm:finite_if_deg_bounded}
  For all $s,q \ge 1$ with $s \leq q$, every separated
  $(s,q)$-critical graph with maximum degree $\Delta$ contains at most $N(q,2\Delta sq)$ edges.
\end{lemma}

Note that the requirement of $s\leq q$ is not a real restriction since the roles of $s$ and $q$ can be swapped by reversing the order of the right endpoints.

\begin{proof}
    Consider a separated $(s,q)$-critical graph $G$ on $m$ edges with $E(G) = [m]$ and maximum degree $\Delta$. 
    Assume for a contradiction that $m > N(q,2\Delta sq)$.

    By definiton, for every edge $i \in [m]$, graph $G-i$ (that is, $G$ after the deletion of $i$) admits an
    $s$-stack $q$-queue layout. We fix such an $s$-stack $q$-queue layout of
    $G-i$ and define a Boolean function $P_i: [m] \setminus \{i\} \rightarrow \{0,1\}$, where
    $P_i(k) = 0$ when $k$ is in a queue and $P_i(k) = 1$ when $k$ is in a stack of the layout.

    First we show that all $P_i$ are $(2\Delta sq)$-close.
    Given $i,j \in [m]$, we suppose for the sake of contradiction that there are $2\Delta sq + 1$
    edges $k \in [m] \setminus \{i,j\}$ with $P_i(k) \neq P_j(k)$. Then without loss of
    generality, there is a set $S$ of size $\Delta sq + 1$ with $P_i(k) = 0$ and $P_j(k) = 1$ for
    all $k \in S$. Since $P_i(k) = 0$ for all $k \in S$, all edges in $S$ are in a queue in the
    fixed $s$-stack $q$-queue layout of $G-i$. Thus $S$ can be covered by $q$ queues. Analogously,
    since $P_j(k) = 1$ for all $k \in S$, $S$ can also be covered by $s$ stacks.
    Now, since $S$ has $\Delta sq+1$ edges and can be covered by $q$ queues, at least one of the
    queues contains $\Delta s + 1$ edges. That is, those $\Delta s+1$ edges pairwise either cross
    or share an endpoint. Observe that $G$ is a bipartite graph of maximum degree $\Delta$, and hence,
    it is $\Delta$-edge-colorable~\cite{Ko16}; that is, $S$ can be partitioned into $\Delta$ matchings.
    By pigeonhole principle, at least one of these matchings contains at least $s+1$ edges. Hence,
    we found $s+1$ pairwise crossing edges in $S$. This contradicts to the
    fact that $S$ can be covered by $s$ stacks, and hence, the assumption that there are $2\Delta sq + 1$
    edges $k \in [m] \setminus \{i,j\}$ with $P_i(k) \neq P_j(k)$. In other words, the
    functions $P_i$ are $(2\Delta sq)$-close.

    Next we apply \cref{lm:Boolean_functions} with $U = E(G) = [m],r = q, d = 2\Delta sq,$ and the assumption that we have more than $N(q,2 \Delta sq)$ edges, from which we obtain a function $P: [m] \rightarrow \{0,1\}$.
    This yields a mixed linear layout for $G$ by using queues for edges assigned to 0 and stacks for edges assigned to 1.
    We aim to show that $q$ queues and $s$ stacks suffice, which gives the desired contradiction since $G$ is $(s,q)$-critical.
    For this, we partition the edges into a set $Q \coloneqq \{k \in [m] \mid P(k) = 0\}$ of edges assigned to queues and a set $S \coloneqq \{k \in [m] \mid P(k) = 1\}$ of edges assigned to stacks. 

    We show that $Q$ contains no $(q+1)$-rainbow.
    Suppose there is a $(q+1)$-rainbow in $Q$, denoted by~$R$. Since $|R| \leq q+1 = r + 1$, by the properties of function $P$, there is an
    $i \in [m] \setminus R$ with $P(k) = P_i(k)$ for all $k \in R$. 
    Since $R \subseteq Q$, we know
    that $P(k) = 0$, and thus $P_i(k) = 0$ for all $k \in R$. 
    That is, all edges from $R$ are in queues, not only in the layout of $G$ but also in the $s$-stack $q$-queue layout of $G-i$.
    However, having a $(q+1)$-rainbow in $q$ queues is a contradiction, so there is no $(q+1)$-rainbow.
    Similarly, since $s+1 \leq q+1 = r + 1$, we conclude that there is no $(s+1)$-twist in set $S$.

    To sum up, we have that $Q$ does not contain a $ (q+1)$-rainbow, so it can be covered by $q$ queues. Likewise, $S$ can be covered by $s$ stacks since the layout is separated.
    Thus, we found an $s$-stack $q$-queue layout of $G$, which contradicts that $G$ is
    $(s,q)$-critical. Therefore, for the number of edges in $G$, it holds that $m \leq N(q,2\Delta sq)$.
\end{proof}

Note that \cref{lm:finite_if_deg_bounded} already yields a finite number of $(s,q)$-critical graphs
of bounded-degree separated graphs, since the size of every such graph is bounded. In order to
extend the result to $k$-critical graphs, we need the following claim.

\begin{lemma}\label{lm:sq_to_k}
	Suppose that the number of edges in an $(s,q)$-critical ordered graph is at most $m(s,q)$ for $s,q \ge 0$.
	Then the number of edges in a $k$-critical ordered graph is at most $\sum\limits_{s+q=k} m(s,q)$.
\end{lemma}

\begin{proof}
	Consider a $k$-critical graph $G$.
	By definition, for any $s + q = k$, $G$ does not admit an $s$-stack $q$-queue layout.
	Thus, for each pair $(s,q)$, $G$ contains an $(s,q)$-critical graph $G_{(s,q)}$ as a subgraph;
    and this subgraph contains at most $m(s,q)$ edges.
	Let $H$ be the subgraph of $G$ with the smallest number of edges that contains all of $G_{(s,q)}$;
	clearly, $H$ has at most $\sum_{s+q=k} m(s,q)$ edges.
	Since $H$ does not admit an $s$-stack $q$-queue layout for any pair $(s,q)$ with $s+q = k$, it follows that
	$H$ does not admit any mixed linear layout on $k$ pages.
	As $G$ is $k$-critical, it cannot have more edges than $H$.
\end{proof}

\cref{lm:sq_to_k,lm:finite_if_deg_bounded} together imply the first part of \cref{thm:critical_sep}.
For the second part of the theorem, we use the next lemma.

\begin{lemma}
	\label{lm:forbidden_degree}
	Let $G$ be a separated $(1, 1)$-critical graph.
	Then $\Delta(G) < 6$.
\end{lemma}
\begin{proof}
    Consider the grid representation of $ G $ in which vertices are represented by
    columns and rows, with a point in the intersection if the two vertices share an edge. Suppose $G$
    has a vertex of degree at least $6$, which is represented, without loss of generality, by a column $ C $.
    Since $ G $ is critical, removing any point in $C$ yields a $1$-stack $1$-queue layout;
    that is, all remaining points can be covered by an increasing sequence (queue) and a decreasing sequence (stack).
    The idea is to consider two such points, $x \in C$ and $y \in C$,
    and obtain two pairs of sequences that can be
    combined into one increasing and one decreasing sequence covering all points. This contradicts to
    $ G $ being $ (1,1) $-critical.

	We refer to \cref{fig:forbidden_degree} for an illustration of how two such layouts are recombined.
    Let $C$ consists of (at least) six points, $x{''}$, $x$, $x'$, $y'$, $y$, $y{''}$,
    in the order from top to bottom. Consider a $1$-stack $1$-queue layout of $G - x$. Observe that
    none of the two sequences contains points both below and above $x$, as otherwise $ x $ could be added, so that the two
    sequences form a $1$-stack $1$-queue layout of $G$. Hence, there are two ways in which
    the two sequences cover the points of $C - x$: either the increasing sequence covers all points above $x$ and the
    decreasing sequence covers everything below $x$, or vice versa. Similarly, removing $y$ from $G$
    yields two possible $1$-stack $1$-queue layouts of $G - y$.

    We cut the grid representations with a horizontal line, $ \ell $, between $ x' $ and $ y' $.
	Take the lower part of the grid representation of $ G-x $, that is, the part that containing $ y $, and
	the upper part of the grid representation of $ G-y $ containing $ x $. 
    Together, these two
	parts cover all edges of $ G $ since no point was removed in the chosen parts. 
    It is left two show that the combined sequences are indeed an increasing and a decreasing sequence. 
    Recall that in the 1-stack 1-queue layouts, above $ x $ and below $ y $ there is at least one point each and they are covered by different sequences. 
    Thus, it holds for the lower part of the combined layout that the increasing sequence hits
	$ \ell $ in column $ C $ or to the left of $ C $, and in the upper part, it continues at $ C $
	or to the right. 
    Similarly, the decreasing sequence starts in the upper part and hits $ \ell $
	at $ C $ or to the left and continues in the lower part at $ C $ or to the right.
\end{proof}

\begin{figure}
    \begin{subfigure}[b]{.19\linewidth}
        \center
        \includegraphics[page=6]{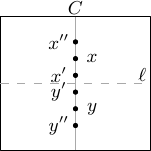}
        \caption{}
    \end{subfigure}
    \hfill
    \begin{subfigure}[b]{.38\linewidth}
        \center
        \includegraphics[page=7]{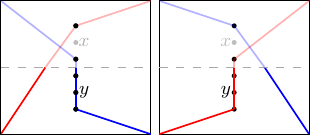}
        \caption{}
    \end{subfigure}
    \hfill
    \begin{subfigure}[b]{.38\linewidth}
        \center
        \includegraphics[page=8]{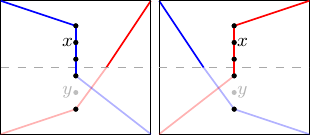}
        \caption{}
    \end{subfigure}
    \caption{
        (a)~A $(1,1)$-critical graph $G$ with a vertex of degree $6$ represented by column $C$.
        (b)~Two options for a $1$-stack (blue decreasing sequence) $1$-queue (red increasing sequence) layout of $G-x$.
        (c)~Two options for a $1$-stack $1$-queue layout of $G-y$.\linebreak
        In all cases the bottom part of the grid below $\ell$ in (b) can be combined with the top part of the grid in (c)
        into a $1$-stack $1$-queue layout of $G$.}
    \label{fig:forbidden_degree}
\end{figure}

By \cref{lm:finite_if_deg_bounded,lm:forbidden_degree}, the number of separated $(1,1)$-critical
graphs is finite.
As the number of $(2,0)$-critical and $(0,2)$-critical graphs is also finite in the separated setting, \cref{lm:sq_to_k} implies that
the number of $2$-critical graphs is bounded.
Thus, the second part of \cref{thm:critical_sep}~follows.

\subsection{Critical Graphs for Non-Separated Layouts}
\label{sec:critical_general}

Based on the literature, we first observe that the number of critical graphs characterizing non-separated $s$-stack $q$-queue layouts and (pure) $s$-stack layouts is expected to be infinite even for matchings, for all $s \geq 4$.
To this end, we use a hardness result stating that it is $\NP$-complete to recognize $4$-stack layouts for matchings with a fixed
layout (from coloring circle graphs)~\cite{Un88}.
Moreover, an already $\NP$-complete mixed linear layout recognition problem with given vertex order remains $\NP$-complete under addition of a stack or queue~\cite{CKN19}.
On the other hand, a finite obstruction set implies
a poly-time recognition, which would contradict the hardness result under the assumption that $ P \neq \NP$.
Similarly, the number of $k$-critical (non-separated) graphs is expected to be infinite due to \cref{thm:mn_hardness}.
In contrast, $(0,q)$-critical graphs are exactly the $(q+1)$-rainbows~\cite{HR92};
thus, there is exactly one $(0,q)$-critical graph for every $q \in \mathbb{N}$.
Therefore, the state-of-the-art does not imply whether the obstruction set for $k$-critical graphs is finite for $ k < 4 $.
In addition, it is interesting to build explicit infinite (though not necessarily complete) collections of
critical graphs for various values of $s, q, k \ge 0$, which we accomplish by the next theorem.

\CriticalGeneralThm*

We start by building an infinite obstruction set for (pure) stack layouts of matchings.
Intuitively, $(2, 0)$-critical graphs can be constructed from odd-length cycles that
require $3$ colors in any proper coloring but removal of a vertex makes the graph $2$-colorable.
What is left to do is realizing the cycle via conflicting edges (that cannot share a stack) of a matching.
The same approach is used for larger values of $s \ge 2$.

\begin{lemma}\label{lm:critical_stacks}
	For every integer $s \ge 2$ and every $n \ge 3$, there exists an $(s, 0)$-critical matching
    graph with at least $n$ vertices.
\end{lemma}
\begin{proof}
	We present a construction of a conflict graph that is $(s, 0)$-critical, i.e., a graph where the nodes are
	the edges of the linear layout and there is an edge between two nodes if the corresponding edges cross.
	The required matching can be easily constructed, see \cref{fig:obstructions_stacks}.

	For $s = 2$, the conflict graph is an odd-length cycle, which clearly cannot be bi-colored.
	It is easy to see that removing any vertex from the graph makes it bi-colorable, and hence, it admits
    a $(2, 0)$-layout.

	For $s \ge 3$, start with a cycle whose length is not a multiple of $s$, say $n = rs + 1$ for some integer $r \ge 1$, and denote the nodes by $ 1, \dots, n $ in cyclic order.
	For every node, add edges $(i, i + j)$ for all $1 \le j < s$ (addition is modulo $n$).
	On the one hand, the graph is not $s$-colorable, since for every $i$, every assignment of $s$ colors to vertices $i, \dots, i+s-2$ (forming an $s$-clique together with the next vertex) forces a unique color for vertex $i+s-1$; the process yields a contradiction with the color of $0$ as $n$ is not a multiple of $s$.
	On the other hand, removing any vertex from the graph makes it $s$-colorable since the described greedy coloring fails only at the last vertex.
	Now the linear layout is obtained from its conflict graph by first placing the left endpoints and then choosing the length of the edges so that they nest above the next $s-1$ left endpoints;
    refer to \cref{fig:obstructions_stacks} for an instance with $s=3$, $r=6$, and $n=20$.
\end{proof}

\begin{figure}
    \center
    \begin{subfigure}[t]{\linewidth}
        \centering
        \includegraphics[page=4]{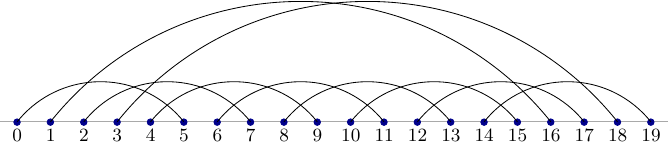}
        \caption{A $(3,0)$-critical graph}
        \label{fig:obstructions_stacks}
    \end{subfigure}
    \hfill
    \begin{subfigure}[t]{\linewidth}
        \centering
        \includegraphics[page=5]{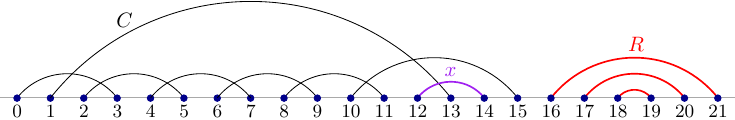}
        \caption{A $2$-critical graph}
        \label{fig:obstructions_mixed}
    \end{subfigure}
    \caption{}
\end{figure}

Before finishing the first part of \cref{thm:critical_nonsep}, we show the base case for the second part, and then lift both parts to prove the theorem.

\begin{lemma}\label{lm:critical_mixed}
    For every $n \ge 3$, there exists a $2$-critical matching with at least $n$ vertices.
\end{lemma}

\begin{proof}
    Let $r \ge 2$ be an even integer. We build a matching with $n = 2(r+2) + 6$ vertices (starting with 0) having
    three types of edges (see \cref{fig:obstructions_mixed}):
    \begin{itemize}
        \item $C = \{(2i, 2i+3), 0 \le i < r-1\} \cup \{(2r-2, 2r+3)\} \cup \{(1, 2r+1)\}$;
        \item a single edge $x = (2r, 2r+2)$;
        \item $R = \{(n-6, n-1)\} \cup \{(n-5, n-2)\} \cup \{(n-4, n-3)\}$.
    \end{itemize}

    We claim that the graph $G_r = ([n-1], C \cup \{x\} \cup R)$ is $2$-critical, that is, $\mn(G_r) = 3$ but
    $\mn(G_r - e) = 2$ for every edge $e$.
    First, we show that $G_r - e$ admits a mixed linear layout on two pages for every edge $e$.
    We consider the three cases $e \in C$, $e = x$, and $ e \in R$ separately.
    \begin{itemize}
        \item If we remove an edge from $C$, then the graph admits a $2$-stack layout:
        The edges of $C \setminus \{e$\} can be assigned to two stacks, while $x$ and edges from $R$
        are assigned to one of the two stacks.

        \item If we remove edge $x$, then the graph admits a $1$-stack $1$-queue layout:
        The ``long'' edge $(1,2r+1)$ of $C$ together with edges from $R$ are assigned to a stack, while
        the ``short'' edges $(2i,2i+3)$ of $C$ are assigned to a queue.

        \item If we remove an edge from $R$, then the graph admits a $2$-queue layout:
        One queue contains an edge from $R$, edge $x$, and the ``long'' edge $(1,2r+1)$ from $C$.
        Another queue contains an edge from $R$ along with the ``short'' edges $(2i,2i+3)$ from $C$.
    \end{itemize}

    Finally, graph $G_r$ does not admit a layout on two (mixed) pages, since on one hand
    (i) it cannot be assigned to two queues ($R$ forms a $3$-rainbow), and
    (ii) it cannot be assigned to two stacks ($C$ is a pattern from \cref{lm:critical_stacks}).
    On the other hand,
    (iii) $G_r$ cannot be assigned to a stack and a queue, since otherwise the ``long'' edge $(1,2r+1)$ of $C$
    is in the queue (it crosses two edges forming a $2$-rainbow) and hence, all edges covered by the
    ``long'' edge would be in a stack, which implies a crossing between two such stack edges.
\end{proof}

Now we are ready to prove the main result of the section.
%
%
\begin{proof}[Proof of \cref{thm:critical_nonsep}]
    We start by considering $k$-critical graphs and show that there exist arbitrarily large $k$-critical matchings.
    The following stronger claim is shown by induction:
	{\it For every $n \ge 3$ and every $k \ge 2$, there exists a matching, denoted $G_k$, with at least $n$ vertices, which (a)~is $k$-critical and (b)~admits a $(k+1)$-stack layout.}

    The base of the induction with $k=2$ is given by \cref{lm:critical_mixed}.
    Assume that $k \ge 2$ and that there is a matching $G_k$ satisfying the induction hypothesis; let us describe how to build $G_{k+1}$.
    To this end, add a $(k+2)$-twist (denoted $M$) nesting above all edges of $G_k$.
	We need to show four claims:
	\begin{itemize}
		\item $G_{k+1}$ does not admit a mixed $(k+1)$-page layout:\\
		Since $M$ is a $(k+2)$-twist, we need at least one queue for $M$ in any $(k+1)$-page layout of $G_{k+1}$.
		That queue cannot be used for an edge of $G_k$; thus, a $(k+1)$-page layout of $G_{k+1}$ exists only if
		a $k$-page layout of $G_k$ exists, a contradiction.

		\item $G_{k+1} - e$ for $e \in G_k$ admits a mixed $(k+1)$-page layout:\\
		Since $G_k$ is $k$-critical, $G_k-e$ admits a $k$-page mixed layout. Use the layout for $G_k-e$ and
		an extra queue for $M$.

		\item $G_{k+1} - e$ for $e \in M$ admits a mixed $(k+1)$-page layout:\\
		$M-e$ is a $(k+1)$-twist, which can be embedded in $k+1$ stacks. By the induction hypothesis, $G_k$ admits a $(k+1)$-stack layout. Since the stacks can be re-used, we have the desired layout.

		\item $G_{k+1}$ admits a $(k+2)$-stack layout, since $M$ can be embedded in so many stacks and $G_k$ admits a $(k+1)$-stack layout.
	\end{itemize}

    A proof for $(s, q)$-critical matchings is similar to the above.
    For a fixed $s \ge 2$,
    the induction is on $q$ with the following hypothesis:
	{\it For every $n \ge 3$, every $s \ge 2$, and every $q \ge 0$, there exists a matching, denoted $G_q$, with at least $n$ vertices, which is $(s, q)$-critical.}

    We fix any $s \geq 2$ and $ n \geq 3 $ and do induction on $q$,
    where the base case with $q=0$ is given by \cref{lm:critical_stacks}.
    Assume that $q \ge 0$ and that there is a matching $G_q$ satisfying the induction hypothesis. 
    We again construct $G_{q+1}$ by adding an $(s+1)$-twist (denoted $M$) nesting above all edges of $G_q$.
    We show three claims:
    \begin{itemize}
        \item $G_{q+1}$ does not admit a mixed $(s, q+1)$-page layout:\\
        Since $M$ is an $(s+1)$-twist, we need at least one queue for $M$ in any layout of $G_{q+1}$ with $s$ stacks.
        That queue cannot be used for an edge of $G_q$; thus, an $(s, q+1)$-page layout of $G_{q+1}$ exists only if
        an $(s, q)$-page layout of $G_q$ exists, a contradiction.

        \item $G_{q+1} - e$ for $e \in G_q$ admits a mixed $(s, q+1)$-page layout:\\
        Since $G_q$ is $(s,q)$-critical, $G_q-e$ admits a mixed $(s,q)$-page layout. Use the layout for $G_q-e$ and
        an extra queue for $M$.

        \item $G_{q+1} - e$ for $e \in M$ admits a mixed $(s, q+1)$-page layout:\\
        $M-e$ is an $s$-twist, which can be embedded in $s$ stacks. By the induction hypothesis, $G_q$ is
        $(s, q)$-critical; hence, it admits an $(s,q+1)$-page layout by assigning one of its edges to a queue.
        Since the $s$ stacks for $M-e$ can be re-used for stacks of the layout of $G_q$, we have the desired $(s, q+1)$-page layout. \qedhere
    \end{itemize}
\end{proof}

Note that while \cref{thm:critical_nonsep} answers the exact version of \cref{op:mn_characterization}
in negative for most values of $s, q$, and $k$, it does not rule out a possibility of a
finite obstruction set for layouts with a small number of pages. In fact, \cref{lm:sq_to_k}
(which holds for non-separated layouts) implies a finite number of $1$-critical graphs, since
there is exactly one $(1, 0)$-critical graph (a $2$-twist) and exactly one $(0, 1)$-critical graph (a $2$-rainbow).
Based on extensive computational experiments, we conjecture that
for $k = 1$, there are exactly eight critical matchings (see \cref{fig:conjk1}), while
for $s = q = 1$, there are exactly twelve critical matchings (see \cref{fig:conj11}).
Note that there also may be further critical graphs that are non-matchings.
%
%
\begin{samepage}
\begin{conjecture}
	\label{conj:11}
	An ordered matching with a fixed (non-separated) layout admits
    \begin{enumerate}[(i)]
        \item a $1$-stack $1$-queue layout if and only if
            it avoids $12$ patterns depicted in \cref{fig:conj11}, and
        \item a mixed layout on $1$ page if and only if
            it avoids $8$ patterns depicted in \cref{fig:conjk1}.
    \end{enumerate}
\end{conjecture}
\end{samepage}



\begin{figure}
    \center
    \begin{subfigure}[t]{\linewidth}
        \centering
        \includegraphics[page=1,width=0.9\linewidth]{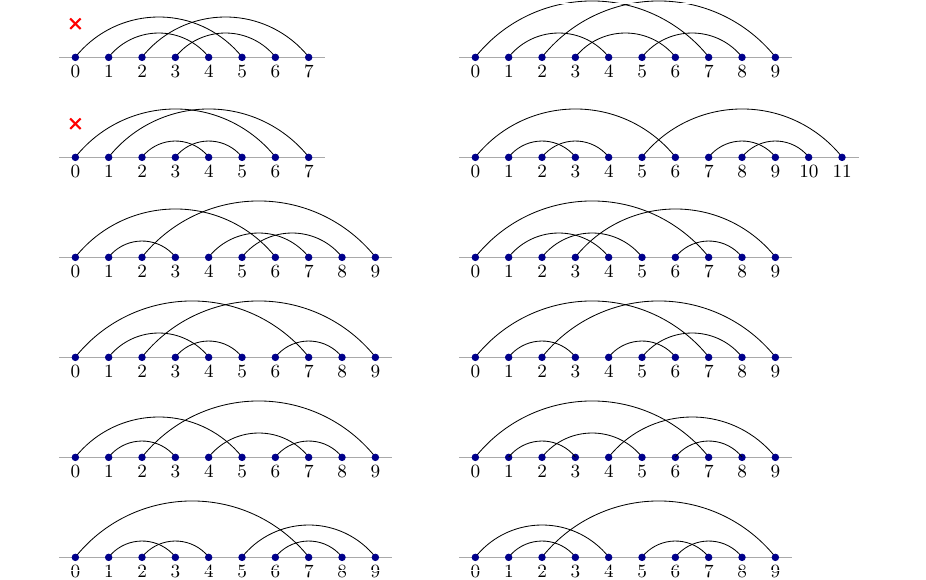}
        \caption{Forbidden patterns for \cref{conj:11}(i). The marked ones are the two 2-thick patterns.}
        \label{fig:conj11}
    \end{subfigure}
    \begin{subfigure}[t]{\linewidth}
        \centering
        \includegraphics[page=2,width=0.9\linewidth]{pics/conjecture}
        \caption{Forbidden patterns for \cref{conj:11}(ii).}
        \label{fig:conjk1}
    \end{subfigure}
    \caption{}
\end{figure}

\section{Conclusions}

In this paper we made the first steps towards characterizing mixed linear layouts of ordered graphs via forbidden patterns.
The most prominent open question for fully resolving \cref{op:mn_characterization} is to transfer
\cref{main:thick} to general graphs with unbounded maximum degree.
As discussed above, the proofs in \cref{sec:general_matching} also work for general ordered graphs;
hence, the challenge is to bound the mixed page number of bipartite graphs in the separated setting.
We expect that thick patterns is the correct choice even for large-degree graphs.

Another interesting question is whether separated mixed linear layouts are characterized by a finite obstruction set;
that is, whether the statement of \cref{thm:critical_sep} holds for graphs with unbounded maximum degree.
Again we expect a positive answer here, and observe that in light of \cref{lm:finite_if_deg_bounded}, it is sufficient to bound
the maximum degree of separated $k$-critical graphs for $k \ge 3$, that is, extend \cref{lm:forbidden_degree} for the case.

Finally, we highlight a possible application of the studied characterizations. Is the mixed page number of
upward planar (possibly, bounded-degree) graphs bounded by a constant?
To answer the question affirmatively for bounded-degree graphs, it suffices to construct a topological order
containing no $k$-thick pattern for some $k \in \mathbb{N}$.

\bibliography{refs_mp}

\end{document}